\documentclass[num-refs]{wiley-article}
\usepackage{wrapfig}

\usepackage{biolinum} 

\usepackage{siunitx}
\usepackage{framed}

\usepackage{stackengine}

\usepackage{autobreak}
\allowdisplaybreaks

\def \C{\mathcal{C}}

\def \E{\mathcal{E}}

\def \G{\mathcal{G}}

\def \M{\mathcal{M}}
\def \N{\mathcal{N}}

\def \S{\mathcal{S}}
\def \T{\mathcal{T}}

\def \W{\mathcal{W}}

\def \bG {\mathcal{G}_{\rm slow}}

\def \sT {\mathscr{T}}

\def \fa{\mathbf{a}}
\def \fb{\mathbf{b}}
\def \fc{\mathbf{c}}
\def \fc{\mathbf{c}}

\def \fg{\mathbf{g}}

\def \fx{\mathbf{x}}
\def \fy{\mathbf{y}}
\def \fz{\mathbf{z}}

\def \fX{\mathbf{X}}
\def \fY{\mathbf{Y}}
\def \fZ{\mathbf{Z}}
\def \f0{\mathbf{0}}

\definecolor{blau_1a}{RGB}{93,133,195}
\definecolor{blau_2a}{RGB}{0,156,218}
\definecolor{gruen_3a}{RGB}{80,182,149}
\definecolor{gruen_4a}{RGB}{175,204,80}
\definecolor{gruen_5a}{RGB}{221,223,72}
\definecolor{orange_6a}{RGB}{255,224,92}
\definecolor{orange_7a}{RGB}{248,186,60}
\definecolor{rot_8a}{RGB}{238,122,52}
\definecolor{rot_9a}{RGB}{233,80,62}
\definecolor{lila_10a}{RGB}{201,48,142}
\definecolor{lila_11a}{RGB}{128,69,151}

\definecolor{blau_1b}{RGB}{0,90,169}
\definecolor{blau_2b}{RGB}{0,131,204}
\definecolor{gruen_3b}{RGB}{0,157,129}
\definecolor{gruen_4b}{RGB}{153,192,0}
\definecolor{gruen_5b}{RGB}{201,212,0}
\definecolor{orange_6b}{RGB}{253,202,0}
\definecolor{orange_7b}{RGB}{245,163,0}
\definecolor{rot_8b}{RGB}{236,101,0}
\definecolor{rot_9b}{RGB}{230,0,26}
\definecolor{lila_10b}{RGB}{166,0,132}
\definecolor{lila_11b}{RGB}{114,16,133}

\definecolor{mycolor1}{rgb}{0.0, 0.18, 0.39}
\definecolor{mycolor2}{RGB}{87,108,67}
\definecolor{mycolor3}{RGB}{8,133,161}
\definecolor{mycolor4}{RGB}{80,91,161}
\definecolor{mycolor5}{RGB}{98,122,157}
\definecolor{mycolor6}{RGB}{255,163,67}
\definecolor{mycolor7}{RGB}{152,205,225}
\definecolor{mycolor8}{RGB}{242,204,48}
\definecolor{mycolor9}{rgb}{0,.5,0}
\definecolor{mycolor10}{rgb}{.59,.44,.09}
\definecolor{mycolor11}{RGB}{231,199,31} 
\definecolor{mycolor12}{RGB}{8,133,161} 
\definecolor{mycolor13}{RGB}{157,188,64} 
\definecolor{mycolor14}{RGB}{194,150,130} 
\definecolor{mycolor15}{RGB}{98,122,157} 
\definecolor{mycolor16}{RGB}{160,160,160} 
\definecolor{mycolor17}{RGB}{115,82,68} 
\definecolor{mycolor18}{RGB}{94,60,108} 
\definecolor{mycolor19}{RGB}{115,82,68} 
\definecolor{mycolor20}{RGB}{255,183,30} 

\usepackage{commath,amsfonts}
  \usepackage{amsmath}
   \usepackage{amssymb}
   \usepackage{bm}

\usepackage{mathtools}
\usepackage{algorithmic}
\usepackage{pgfplots} 
\usetikzlibrary{arrows}
\pgfplotsset{compat=1.15}
\usepackage{graphicx}
\usepackage{xfrac}
\usepackage{colortbl}
\usepackage{cancel}
\usepgfplotslibrary{fillbetween}
\usepackage{float}
\usepackage{hyperref}
\usepackage{multirow}
\usepackage{xcolor}
\usepackage{mathrsfs}
\usepackage{bbm}
\usepackage{autobreak}
\allowdisplaybreaks
\usepackage{tikz}
\usetikzlibrary{arrows.meta}
\usepackage{amsbsy}
\usepackage{fixmath}
\usepackage{silence}
\usetikzlibrary{math}
\usepackage{graphicx}
\usepackage{pgfgantt}
\usepackage{pdflscape}

\usepackage{upgreek}
\usepackage{nicematrix}

\usepackage{leftindex}
\let\svtikzpicture\tikzpicture
\def\tikzpicture{\noindent\svtikzpicture}

\papertype{ArXiv Preprint}
\paperfield{Information Theory}

\title{Deterministic $K$-identification For Slow Fading Channel}

\author[1]{Mohammad Javad Salariseddigh}
\author[2]{Muris Spahovic}
\author[3]{Christian Deppe}

\affil[1]{Institute For Communication Engineering, Technical University of Munich, Munich, 80333, Germany, E-mail: mjss@tum.de}
\affil[2]{Institute For Communication Engineering, Technical University of Munich, Munich, 80333, Germany, E-mail: ge25ron@mytum.de}
\affil[3]{Institute For Communication Engineering, Technical University of Munich, Munich, 80333, Germany, E-mail: christian.deppe@tum.de}

\fundinginfo{
\vspace{.8mm}
\begin{enumerate}
    \item[1] 6G-Life Project, GNT: 16KISK002
    \item[3] BMBF, GNT: 16KIS1005, 16KISQ028;\\6G-Life Project, GNT: 16KISK002
\end{enumerate}
}

\begin{document}

\begin{frontmatter}
\maketitle

\begin{abstract}
Deterministic $K$-identification (DKI) is addressed for Gaussian channels with slow fading (GSF), where the transmitter is restricted to an average power constraint and channel side information is available at the decoder. We derive lower and upper bounds on the DKI capacity when the number of identifiable messages $K$ may grow sub-linearly with the codeword length $n$. As a key finding, we establish that for deterministic encoding, assuming that the number of identifiable messages $K = 2^{\kappa \log n}$ with $\kappa \in [0,1)$ being the identification target rate, the codebook size scales as $2^{(n\log n)R}$, where $R$ is the coding rate.
\end{abstract}
\end{frontmatter}
 
\section{Introduction}
Modern communications within the scope of future generation wireless networks (XG) \cite{6G+,6G_PST} require the transfer of extensive amount of data in wireless communication, including smart applications for internet of things \cite{DHV17}, cellular communication, sensor networks, etc. One of the basic and abstract models for wireless communication systems is the fading channel \cite{LYHW18,TV05}. Unlike the fast fading setting, where the coherence time of the channel is small relative to the latency requirement of the application \cite{Biglieri98,TV05}, in the slow fading regime, the latency is short compared to the coherence time \cite{Biglieri98,TV05}. In some appliances, the receiver may acquire channel side information (CSI) by instantaneous estimation of the channel parameters \cite{GV97,AF12}.

Numerous applications of future generation wireless networks (XG) \cite{6G+,6G_PST} are linked with event-triggered communication systems. In such systems, Shannon's message transmission capacity, as studied in \cite{S48}, is not the appropriate metric for the performance evaluation, instead, the identification capacity is deemed to be an essential quantitative measure. In particular, for object-finding or event-detection scenarios, where the receiver aims to determine the presence of an object or determine the occurrence of an specific event in terms of a reliable Yes\,/\,No answer, the so-called identification capacity is the key applicable performance measure \cite{AD89}. 

The original coding scheme for the identification problem introduced by Ahlswede and Dueck \cite{AD89} employs a randomized encoder, where the codewords are tailored according to distributions. The codebook size for randomized identification (RI) grows double-exponentially in the codeword length $n$, i.e., $\sim 2^{ 2^{nR}}$ \cite{AD89}, where $R$ is the coding rate. Realization of RI codes entails high complexity and is challenging for the applications; cf. \cite{Salariseddigh22}. Ahlswede and Dueck were inspired to introduce RI by the work of J\'aJ\'a \cite{J85} who considered an deterministic identification (DI) in communication complexity \cite{Yao79,LL89,AA20}. This problem can be also considered in the communication setting. Here, the codewords are selected via a deterministic function from the messages. DI may be preferred over RI in complexity-constrained applications of MC systems, where the generation of random codewords could be challenging. The DI for discrete memoryless channels (DMCs) with average power constraint is studied in \cite{Salariseddigh_ICC,Salariseddigh_IT} where the codebook size grows exponentially in the codeword length \cite{AD89,Salariseddigh_ICC}. Furthermore, the DI for continuous alphabet channels including Gaussian channels with fast and slow fading and the memoryless discrete-time Poisson channel (GSF) is addressed in \cite{Salariseddigh_ITW,Salariseddigh_IT,Salariseddigh_GC_IEEE} where a new observation regarding the codebook size is reported, namely, it scales super-exponentially with the codeword length $n$, i.e., $\sim 2^{(n\log n)R}$.

In the (standard) identification problem \cite{AD89}, the receiver is interested in a \emph{single} message which we refer to as the target message in the rest of the paper. However, for the $K$-identification problem \cite{AH08G}, the receiver aims to determine the presence of a single message within a set of messages referred to as the target message set\footnote{\,\textcolor{mycolor5}{For instance, the $K$-identification scenario may be used whenever a person aims to determine whether a winner is among their favourite teams or within the context of lottery prize; when people seek to know if a lottery number is among their collection of numbers.}}. The $K$-identification scenario may be understood as the generalization of the original identification problem within this interpretation: the target message (singleton) is substituted with a set of more than one element with size $K$. The first result for $K$-identification is derived by Ahlswede for a DMC $\W$ with randomized encoder setting as follows: Assume that $K=2^{\kappa n}$, then the set of all achievable coding and target identification rate pairs, i.e., $(R,\kappa)$ with a codebook of double exponentially large, i.e., $M=2^{2^{nR}}$, contains $\left\{ (R, \kappa) : 0 \leq R, \kappa; R + 2\kappa \leq \mathbb{C}_{RI}(\W,M,K) \right\}$; see \cite[Th.~1]{AH08G}. To the best of the authors' knowledge, the fundamental performance limits of DKI for the Gaussian channels has not been studied in the literature, yet.
\subsection{Contributions}
In this paper, we consider identification systems employing deterministic encoder and receivers that are interested to accomplish the $K$-identification task, namely, finding an object in a target message set of size $K$ where $K=2^{\kappa \log n}$ for $\kappa \in [0,1)$ scales sub-linearly in the codeword length $n$. We assume that the noise is additive Gaussian and the signal experiences slow fading process. Further, we assume that the channel side information (CSI) is available at the decoder. We formulate the problem of DKI over the GSF under average power constraint which account for the restricted signal energy in the transmitter. As our main objective, we investigate the fundamental performance limits of DKI over the slow fading channel. In particular, this paper makes the following contributions:
\begin{itemize}
    \item[$\blacklozenge$] \textbf{\textcolor{blau_2b}{Generalized Identification Model}}: In several identification systems, often the size of target message set $K$ can be large, particularly when one by one comparison is not demanded due to the delay constraint. In addition, the value of $K$ may increases with the codeword lengths $n$. To do so, we consider a generalized identification model that captures the standard channel (i.e., $K=1$), identification channels with constant $K>1$, and identification channels for which $K$ increases with the codeword length $n$. To the best of the authors' knowledge, such a generalized deterministic identification model has not been studied in the literature, yet.
    \item[$\blacklozenge$] \textbf{\textcolor{blau_2b}{Codebook Scale}}: We establish that the codebook size of DKI problem over the Gaussian channels with slow fading for deterministic encoding scales in $n$ similar to the DI problem ($K=1$) \cite{Salariseddigh_ITW,Salariseddigh_IT}, namely super-exponentially in the codeword length ($\sim 2^{(n\log n)R}$), even when the size of target message set scale as $K=2^{\kappa \log n}$ for any $\kappa \in [0,1)$, which we refer to as the target identification rate. This observation suggests that increasing the number of target messages does not change the scale of the codebook derived for DI over the Gaussian channels \cite{Salariseddigh_ITW}.
    \item[$\blacklozenge$] \textbf{\textcolor{blau_2b}{Capacity Bounds}}: We derive DKI capacity bounds for the slow fading channel with constant $K \geq 1$ and growing size of the target message set $K = 2^{\kappa \log n}$, respectively. We show that for constant $K$, the proposed lower and upper bounds on $R$ are independent of $K$, whereas for growing number of target messages, they are functions of the target identification rate $\kappa$.
    \item[$\blacklozenge$] \textbf{\textcolor{blau_2b}{Technical Novelty}}: To obtain the proposed lower bound, the existence of an appropriate sphere packing within the input space, for which the distance between the centers of the spheres does not fall below a certain value, is guaranteed. This packing incorporates the effect of number of target messages as a function of $\kappa$. In particular, we consider the packing of hyper spheres inside a larger large hyper sphere, whose radius grows in both the codeword length $n$ and the target identification rate $\kappa$, i.e., $\sim n^{\frac{1+\kappa}{4}}$. For derivation of the upper bound, we assume that for given sequences of codes with vanishing error probabilities, a certain minimum distance between the codewords is asserted, where this distance depends on the target identification rate and decreases as $K$ grows.
\end{itemize}
\subsection{Organization}
The remainder of this paper is structured as follows. In Section~\ref{Sec.SysModel}, system model is explained and the required preliminaries regarding DKI codes are established. Section~\ref{Sec.Res} provides the main contributions and results on the message $K$-identification capacity of the slow fading channel. Finally, Section~\ref{Sec.Summary} of the paper concludes with a summary and directions for future research.

\subsection{Notations}
We use the following notations throughout this paper: Blackboard bold letters $\mathbbmss{K,X,Y,Z},\ldots$ are used for alphabet sets. Lower case letters $x,y,z,\ldots$ stand for constants and values of random variables, and upper case letters $X,Y,Z,\ldots$ stand for random variables. The set of consecutive natural numbers from $1$ through $M$ is denoted by $[\![M]\!]$. The set of whole numbers is denoted by $\mathbb{N}_{0} \triangleq \{0,1,2,\ldots\}$. The set of real and non-negative numbers are denoted by $\mathbb{R}$ and $\mathbb{R}_{+}$, respectively. The distribution of a real random variable $X$ is specified by a cumulative distribution function (cdf) $F_X(x) = \Pr(X\leq x)$ for $x\in\mathbb{R}$, or alternatively, by a probability density function (pdf) $f_X(x)$, when it exists. Lower case bold symbol $\fx$ and $\fy$ stand for row vectors. A random sequence $\fX$ and its distribution $F_{\fX}(\fx)$ are defined accordingly. All logarithms and information quantities are for base $2$. The gamma function for non-positive integer $x$ is denoted by $\Gamma(x)$ and is defined as $\Gamma (x) = (x-1) !$, where $(x-1)! \overset{\text{\tiny def}}{=} (x-1) \times (x-2) \times \dots \times 1$. The $\ell_2$-norm and $\ell_{\infty}$-norm of vector $\fx$ are denoted by $\norm{\mathbf{x}}$ and $\norm{\mathbf{x}}_{\infty}$, respectively. Furthermore, we denote the $n$-dimensional hyper sphere of radius $r$ centered at $\fx_0$ with respect to the $\ell_2$-norm by $\S_{\fx_0}(n,r) = \{\fx\in\mathbb{R}_{+}^n : \norm{\fx-\fx_0} \leq r \}$. We use $\mathbf{0} = (0,\ldots,0)$ to represent coordination of the origin. The closure of a set $A$ is denoted by $\text{cl}(A)$. We denote the GSF with $K$ number of target messages by $\bG$.
\section{System Model and Preliminaries}
\label{Sec.SysModel}
In this section, we present the adopted system model and establish some preliminaries regarding DKI coding.
\subsection{System Model}
We consider an identification-focused communication setup, where the decoder seeks to accomplish the following task: Determining whether or not an specific message belongs\footnote{\,\textcolor{mycolor5}{We assume that the transmitter does not know which specific $K$ messages the decoder is interested in. This assumption is justified by the fact that otherwise, entire communication setting is specialized to transmission of only one indicator bit between Alice and Bob.}} to a set of messages called target message set; see Figure~\ref{Fig.E2E_Chain}. We assume that the signal experiences an additive Gaussian noise and slow fading process.
\begin{figure}[H]
    \centering
	\scalebox{1.05}{
\tikzstyle{farbverlauf} = [ top color=white, bottom color=white!80!gray]
\tikzstyle{l} = [draw, -latex']
\tikzstyle{block1} = [draw, top color = white, middle color = cyan!30, rectangle, rounded corners, minimum height=2em, minimum width=2.5em]
\tikzstyle{block2} = [draw, top color = white, middle color = cyan!30, rectangle, rounded corners, minimum height=2em, minimum width=2.5em]
\tikzstyle{block3} = [draw, top color = white, middle color = orange_6b!30, rectangle, rounded corners, minimum height=1em, minimum width=2em]
\tikzstyle{input} = [coordinate]
\tikzstyle{sum} = [draw, circle,inner sep=0pt, minimum size=5mm,  thick]
\tikzstyle{product} = [draw, circle,inner sep=0pt, minimum size=5mm,  thick]
\tikzstyle{arrow}=[draw,->]
\begin{tikzpicture}[auto, node distance=2cm,>=latex']
\node[] (M) {$i$};
\node[block1, right=.5cm of M] (enc) {Enc};
\node[product, right=1.3cm of enc] (channel1) {$\times$};
\node[block3, above=1.1cm of channel1] (fade) {$f_G$};
\node[sum, right=1.3cm of channel1] (channel2) {$+$};
\node[block2, right=1cm of channel2] (dec) {Dec};
\node[right=.5cm of dec] (Output) {\text{\small Yes\,/\,No}};
\node[below=.5cm of dec] (Target) {$j_1,\ldots,j_K$};
\node[above=.7cm of channel2] (noise) {$Z_t$};
\draw[->] (M) -- (enc);
\draw[->] (enc) --node[above]{$\small c_{i,t}$} (channel1);
\draw[->] (channel1) -- (channel2);
\draw[->] (fade) --node[right]{$G$} (channel1);
\path [dashed,l] (fade) -| node [text width=2.5cm,above] {$G$} (dec) ;
\draw[->] (noise) -- (channel2);
\draw[->] (channel2) --node[above]{$Y_t$} (dec);
\draw[->] (Target) --(dec);
\draw[->] (dec) -- (Output);
\end{tikzpicture}
}
	\caption{End-to-end transmission chain for DKI communication in a wireless communication system modelled as a GSF. The transmitter maps message $i$ onto a codeword $\fc_i = (c_{i,1},\ldots,c_{i,n})$. The receiver is provided with an arbitrary target message set $\mathbbmss{K} = \{j_1,\ldots,j_K\}$, and given the channel output vector $\fY$, it asks whether the sent message $i$ belong to set of $K$ messages $\{j_1,\ldots,j_K\}$ or not.}
    \label{Fig.E2E_Chain}
\end{figure}
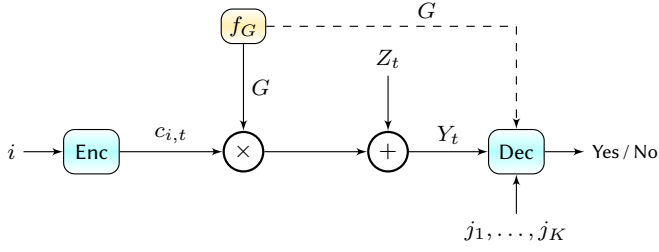
To attain this objective, a coded communication between the transmitter and the receiver over $n$ channel uses of a Gaussian channel with slow fading is established\footnote{\,\textcolor{mycolor5}{The proposed performance bounds works regardless of whether or not an specific code is used for communication, although proper codes may be required to approach such performance limits.}}. We consider the slow fading channel $\bG$ which arises as a channel model in the context of wireless communication \cite{TV05} where the input-output relation is given by
\begin{align}
    Y_t = Gx_t + Z_t \;,\,
\end{align}
where $G_t = G \sim f_{G}$ is a continuous random variable $\sim f_G(g)$, and the noise sequence $\bar{\fZ} \overset{\text{\scriptsize i.i.d.}}{\sim} \mathcal{N}\big(0,\frac{\sigma_Z^2}{n}\big)$ where $\sigma_Z^2>0$ is bounded away from zero. We assume that $G$ has finite expectation and variance $\text{var}(G)>0$. Further, assume that the values of $G$ belong to a set $\G$ where $\gamma \overset{\text{\tiny def}}{=} \underset{G \in \G}{\inf} \, |G| $, that is, the set $\G$ has a constant infimum or equivalently, the fading coefficients are bounded away from zero, i.e., $|G_t| > \gamma \,, \forall t \in [\![n]\!]$ with probability $1$.

The average power constraint on the codewords is
\begin{align}
    \label{Ineq.Const_X}
    \frac{1}{n} \norm{\fx}^2 \leq P_{\,\text{avg}} \,,\,
\end{align}
, where $P_{\,\text{avg}} > 0$ constrain the energy of codeword over the entire $n$ channel uses.
\subsection{DI Coding For The GSF}
The definition of a DKI code for the GSF $\bG$ is given below.
\begin{definition}[Slow Fading DKI code]
\label{Def.ISI-Poisson-Code}
An $(n,\allowbreak M(n,R),\allowbreak K(n,\allowbreak \kappa), \allowbreak e_1, \allowbreak e_2)$ DKI code for a GSF $\bG$ under average power constraint of $P_{\,\text{ave}}$, and for integers $M(n,R)$ and $K(n,\kappa)$, where $n$ and $R$ are the codeword length and coding rate, respectively, with CSI at the decoder is defined as a system $(\C,\sT_{\mathbbmss{K}})$, which consists of a codebook $\C=\{ \mathbf{c}_i \}_{i\in[\![M]\!]} \subset \mathbb{R}^n$, such that
\begin{align}
    \frac{1}{n} \norm{\fc_i}^2 \leq P_{\,\text{avg}} \,,\,
\end{align}
$\forall i\in[\![M]\!]$ and a decoder
\begin{align}
    \label{Eq.DefUnion}
    \sT_{\mathbbmss{K}} = \bigcup_{j \in \mathbbmss{K}} \mathbbmss{T}_{j,g} \;,\,
\end{align}
where $\mathbbmss{T}_{j,g} \subset \mathbb{R}^n$, for $j \in [\![M]\!]$, $g \in \G$, and $\mathbbmss{K} \in \binom{M}{K}$\footnote{\,\textcolor{mycolor5}{We recall that $\binom{M}{K}$ is the family of all subsets of $[\![M]\!]$ with size $K$ and DKI code definition applies to every possible choice of set $\mathbbmss{K}$ with $K$ arbitrary} \textcolor{mycolor5}{messages from the original message set $[\![M]\!]$.}}.
\begin{figure}[!tb]
\centering
\scalebox{1}{
\begin{tikzpicture}[scale=.6][thick]

\draw[dashed] (0.07,.5) circle (3.5cm);

\draw (0,3.6) circle (.1cm);
\draw (3,1.5) circle (.1cm);
\draw (2,-.1) circle (.1cm);
\draw (-2,2) circle (.1cm);
\draw (-.1,.6) circle (.1cm);
\draw (1.4,-1.4) circle (.1cm);
\draw (-1.8,-.8) circle (.1cm);

\node at (0,3.2) {$\fc_2$};
\draw[dashed] (0,3.6) circle (0.2cm);
\draw [fill=mycolor4, fill opacity=0.7] (0,3.6) circle (.1cm);

\node at (3,1.1) {$\fc_3$};
\draw[dashed] (3,1.5) circle (0.2cm);
\draw [fill=mycolor4, fill opacity=0.7] (3,1.5) circle (.1cm);

\node at (2,-.5) {$\fc_4$};

\node at (-2,1.6) {$\fc_1$};
\node at (-.1,.2) {$\fc_5$};

\node at (1.4,-1.8) {$\fc_6$};
\draw[dashed] (1.4,-1.4) circle (.2cm);
\draw [fill=mycolor4, fill opacity=0.7] (1.4,-1.4) circle (.1cm);

\node at (-1.8,-1.2) {$\fc_7$};

\draw[->] (0,4.65) -- ++(0,1.5)  node [fill=white,inner sep=3pt](a){$\text{input space}$};

\draw[->] (-2.3+15.07,4.55) -- ++(0,1.5)  node [fill=white,inner sep=3pt](a){output space};

\draw[dashed] (-0.3+13.07,.4) circle (3.5cm);

\draw (-0.3+14.50,-.5) circle (1.2cm); 
\draw [fill=gray!40, fill opacity=0.7] (-0.3+14.50,-.5) circle (1.2cm);
\node at (-0.1+14.50,-.5) {$\mathbbmss{T}_{5}$};

\draw (-4.3+15.50,1.5) circle (1.2cm);    
\draw [fill=gray!40, fill opacity=0.7] (-4.3+15.50,1.5) circle (1.2cm);
\node at (-4.3+15.30,1.5) {$\mathbbmss{T}_{1}$};

\draw (-2.3+15.07,-1.4) circle (1.2cm);    
\draw [fill=gray!40, fill opacity=0.7] (-2.3+15.07,-1.4) circle (1.2cm);
\node at (-2.3+15.07,-1.6) {$\mathbbmss{T}_{6}$};

\draw (-3.3+14.50,-.4) circle (1.2cm);    
\draw [fill=gray!40, fill opacity=0.7] (-3.3+14.50,-.4) circle (1.2cm);
\node at (-3.3+14.30,-.6) {$\mathbbmss{T}_{7}$};

\draw (-2.3+15.07,2.4) circle (1.2cm);     
\draw [fill=cyan!40, fill opacity=0.7] (-2.3+15.07,2.4) circle (1.2cm);
\node at (-2.3+15.07,2.6) {$\mathbbmss{T}_{2}$};

\draw (-0.3+14.50,1.4) circle (1.2cm);     
\draw [fill=cyan!40, fill opacity=0.7] (-0.3+14.50,1.4) circle (1.2cm);
\node at (-0.3+14.70,1.4) {$\mathbbmss{T}_{3}$};

\draw (-2.4+15.07,.5) circle (1.2cm);    
\draw [fill=cyan!40, fill opacity=0.7] (-2.4+15.07,.5) circle (1.2cm);
\node at (-2.4+15.07,.5) {$\mathbbmss{T}_{4}$};

\path (0.2,3.7) edge [-> , thick, mycolor9, bend left] node [sloped,midway,above,font=\small] {correct identification}(13,3);

\path (3.2,1.6) edge [-> , thick, rot_8b, bend left] node [sloped,midway,above,font=\small] {type I error}(11.0,2.0);

\path (1.6,-1.3) edge [-> , thick, rot_9b, bend left] node [sloped,midway,above,font=\small] {type II error}(-3.4+16,0.1);

\end{tikzpicture}
}
\captionsetup{justification=justified}
\caption{Illustration of a deterministic 3-identification setting with target message set $\mathbbmss{K} = \{2,3,4\}$. In the correct identification scenario, channel output is observed in the union of individual decoder $\mathbbmss{T}_{j,g}$ where $j$ belongs to the target message set. Type I error occurs if the channel output is detected in the complement of union of individual decoders for which the index of codeword at the lest belongs to. The case where the index of codeword at the left does not match to any of the individual decoders for which the channel output belongs to the their union, is referred to as the type II error.}
\label{Fig.DKI-Code}
\end{figure}
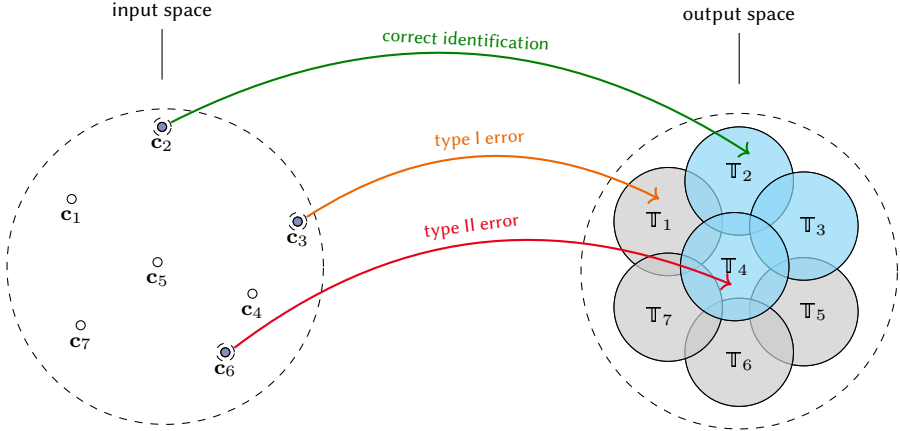
Given a message $i\in [\![M]\!]$, the encoder transmits $\mathbf{c}_i$, and the decoder's aim is to answer the following question: Was a desired message $j \in \mathbbmss{K}$ sent or not? There are two types of errors that may occur (see Figure~\ref{Fig.DKI-Code}): Rejection of the true message for $i\in \mathbbmss{K}$ (type I), or acceptance of a false message for $i \notin \mathbbmss{K}$ (type II). The corresponding error probabilities of the DKI code $(\C,\sT_{\mathbbmss{K}})$ are given by
\begin{align}
    \label{Eq.TypeIError}
    & P_{e,1}(i) = \sup_{g\in\G} \left[ \Pr \left( \fY \in \sT_{\mathbbmss{K}}^c \,\big|\, \fx = \fc_i \right) \right]_{i \in \mathbbmss{K}} = \sup_{g\in\G} \left[ 1 - \int_{\sT_{\mathbbmss{K}}} f_{\fZ} (\fy - g \fc_i) d\fy \right]_{i \in \mathbbmss{K}}
    \\
    & P_{e,2}(i,\mathbbmss{K}) = \sup_{g\in\G} \left[ \Pr \left( \fY \in \sT_{\mathbbmss{K}} \,\big|\, \fx = \fc_i \right) \right]_{i \notin \mathbbmss{K}} =
    \sup_{g\in\G} \left[\int_{\sT_{\mathbbmss{K}}} f_{\fZ} (\fy - g \fc_i) d\fy \right]_{i \notin \mathbbmss{K}}
    \label{Eq.TypeIIError}
\end{align}
where
\begin{align}
    f_{\fZ}(\fz) & = f_{\fZ}(\fy - g \fc_i)
    \nonumber\\
    & = \prod_{t=1}^n f_{\fZ}(y_t-gc_{i,t})
    \nonumber\\&
    = \prod_{t=1}^n \frac{1}{(2\pi\sigma_Z^2)^{1/2}} e^{-z_t^2/2\sigma_Z^2}
    \nonumber\\&
    = \frac{1}{(2\pi\sigma_Z^2)^{n/2}} e^{-\norm{\fz}^2/2\sigma_Z^2}
    \,,\,
\end{align}
(see Figure~\ref{Fig.E2E_Chain}) and satisfy the following bounds $P_{e,1}(i) \leq e_1 \,,\, \forall i \in \mathbbmss{K}$ and $P_{e,2}(i,\mathbbmss{K}) \leq e_2 \,,\, \forall i \notin \mathbbmss{K}$, where $\mathbbmss{K} \in \binom{M}{K}$ and every $e_1, e_2>0$.

A rate $R>0$ is called achievable if for every $e_1, \allowbreak e_2>0$ and sufficiently large $n$, there exists an $(n,\allowbreak M(n\allowbreak,R),\allowbreak K(n,\allowbreak \kappa), \allowbreak e_1, \allowbreak e_2)$ DKI code. The DKI capacity of the GSF $\bG$ is defined as the supremum of all achievable rates, and is denoted by $\mathbb{C}_{DI}(\bG,M,K)$.
\end{definition}
\begin{remark}
\label{FadingCoeffsZero}
If the fading coefficients can be zero or arbitrarily close to zero, i.e., $0\in \text{cl}(\G)$, then it immediately
follows that the DKI capacity is zero. To see this, observe that if $0\in \text{cl}(\G)$, then
\begin{align}
     P_{e,1}(i) + P_{e,2}(i,\mathbbmss{K}) & = \sup_{g\in\G} \left[ 1- \int_{\sT_{\mathbbmss{K}}} f_{\fZ} (\fy - g \fc_i) \, d\fy \right] + \sup_{g\in\G} \left[ \int_{\sT_{\mathbbmss{K}}} f_{\fZ} (\fy - g \fc_i) \, d\fy \right]
     \nonumber\\&
     \geq \left[ 1- \int_{\sT_{\mathbbmss{K}}} f_{\fZ} (\fy - g \fc_i) \, d\fy \right]_{\substack{g=0,\\i \in \mathbbmss{K}}} + \left[ \int_{\sT_{\mathbbmss{K}}} f_{\fZ} (\fy - g \fc_i) \, d\fy
     \right]_{\substack{g=0,\\i \notin \mathbbmss{K}}}
     \nonumber\\&
     = 1 \;.\,
     \label{Eq.FadingCoeffsZero}
\end{align}
\end{remark}
\section{DKI Capacity of the GSF}
\label{Sec.Res}

In this section, we first present our main results, i.e., lower and upper bounds on the achievable identification rates for the GSF. Subsequently, we provide the detailed proofs of these bounds.

\subsection{Main Results}
The DKI capacity theorem for GSF $\bG$ is stated below.
\begin{theorem}
\label{Th.DKI-Capacity}
Consider the GSF $\bG$ and assume that the fading coefficients are bounded away from zero, i.e., $0 \notin \text{cl}(\G)$. Further, assume that the number of target messages scales sub-linearly with codeword length $n$, i.e., $K(n,\kappa) = 2^{\kappa \log n}$, where $\kappa \in [0,1)$. Then the DKI capacity of $\bG$ subject to average power constraint of the form $\norm{\fc_i}^2 \leq nP_{\,\text{ave}}$ and a codebook of super-exponential scale, i.e., $M(n,R)=2^{(n\log n)R}$, is bounded by
\begin{align}
    \label{Ineq.LU}
    \frac{1-\kappa}{4} \leq \mathbb{C}_{DI}(\bG,M,K) \leq 1 + \kappa \,.\,
\end{align}
\end{theorem}
\begin{proof}
The proof of Theorem~\ref{Th.DKI-Capacity} consists of two parts, namely the achievability and the converse proofs, which are provided in Sections~\ref{Sec.Achiev} and \ref{App.Conv}, respectively.
\end{proof}

\begin{remark}
    The result in Theorem~\ref{Th.DKI-Capacity} comprises the following three special cases in terms of $K$:
    \begin{itemize}
        \item[$\blacklozenge$ \; \textbf{\textcolor{blau_2b}{Unit $K=1$}}:] This cases accounts for a standard identification setup ($\kappa = 0$), that is, when the target message set is a degenerate case $\mathbbmss{K} = \{i\}_{i \in [\![M]\!]}$, i.e., $K = |\mathbbmss{K}| = 1$. Therefore, the identification setup as studied in \cite{AD89} can be regarded as a special case of $K$-identification. This result is known in the identification literature \cite{Salariseddigh_ITW,Salariseddigh_arXiv_ITW,Salariseddigh_IT,AD89}.
        \item[$\blacklozenge$ \; \textbf{\textcolor{blau_2b}{Constant $K>1$}}:] Constant $K>1$ implies $\kappa \to 0$ as $n\to\infty$. Surprisingly, our capacity result in Theorem~\ref{Th.DKI-Capacity} reveals that the bounds for the GSF with constant finite $K>1$ are in fact identical to those for the memoryless GSF given in \cite{Salariseddigh_ITW,Salariseddigh_arXiv_ITW,Salariseddigh_IT}.
        \item[$\blacklozenge$ \; \textbf{\textcolor{blau_2b}{Growing $K$}}:] Our capacity results reveal that reliable identification is possible even when $K$ scales with the codeword length as $\sim 2^{\kappa \log n}$ for $\kappa \in [0,1)$. Moreover, the impact of target identificaiton rate $\kappa$ is reflected in the capacity lower and upper bounds in \eqref{Ineq.LU}, where the bounds respectively decrease and increase in $\kappa$.
\end{itemize}    
\end{remark}
\subsection{Achievability}
\label{Sec.Achiev}
The achievability proof consists of the following two main steps.
\begin{itemize}
    \item \textbf{Step 1:} We propose a codebook construction and  derive an analytical lower bound on the corresponding codebook size using inequalities for sphere packing density.
    \item \textbf{Step~2:} To prove that this codebook leads to an achievable rate, we propose a decoder and show that the corresponding type I and type II error rates vanished as $n \to \infty$.
\end{itemize}
\subsubsection{Normalization}
Since the decoder can normalize the output symbols by $\sqrt{n}$, we have an equivalent input-output relation,
\begin{align}
    \bar{Y}_t = G\bar{x}_t + \bar{Z}_t \;,\,
\end{align}
where $G_t=G$ $\sim f_{G}$, and the noise sequence $\Bar{\fZ} \overset{\text{\scriptsize i.i.d.}}{\sim} \mathcal{N}\big(0,\frac{\sigma_Z^2}{n}\big)$, with an input power constraint
\begin{align}
    \label{Ineq.Ave_Pow_Constr}
    \norm{\bar{\fx}} \leq \sqrt{A} \;,\,
\end{align}
where $A \overset{\text{\scriptsize def}}{=} P_{\,\text{ave}}$ and
\begin{align}
    & \bar{\fx}=\frac{1}{\sqrt{n}}\fx \hspace{4mm} \,,\, \hspace{4mm} \bar{\fZ}=\frac{1}{\sqrt{n}}\fZ \hspace{4mm} \,,\, \hspace{4mm} \bar{\fY}=\frac{1}{\sqrt{n}}\fY \,.
\end{align}
\subsubsection*{Codebook Construction}
We use a packing arrangement of non-overlapping hyper spheres of radius $r_0 = \sqrt{\theta_n}$ in a large hyper sphere with radius $\sqrt{A}-\sqrt{\theta_n}$, with
\begin{align}
    \theta_n = \frac{A\sqrt{K}}{n^{\frac{1}{2}(1-b)}} = \frac{A}{n^{\frac{1}{2}(1-(b+\kappa))}} \;,\,
\end{align}
where $0 < b < 1$ is an arbitrarily small constant\footnote{\,\textcolor{mycolor5}{we recall that our achievability proof works for any $b\in(0,1)$; however, arbitrarily small values of $b$ are of interest since they result in the tightest lower bound.}}, and $\kappa \in [0,1)$.

Let $\mathscr{S}$ denote a sphere packing, i.e., an arrangement of $M$ non-overlapping spheres $\S_{\bar{\fc}_i}(n,r_0)$, $i\in [\![M]\!]$, that are packed inside the larger sphere $\S_{\f0}(n,\sqrt{A}-\sqrt{\theta_n})$ with radius $\sqrt{A}-\sqrt{\theta_n}$. As opposed to standard sphere packing in coding techniques \cite{CHSN13}, the spheres are not necessarily entirely contained within the larger sphere. That is, we only require that the centers of the spheres are inside $\S_{\f0}(n,\sqrt{A}-\sqrt{\theta_n})$ and are disjoint from each other and have a non-empty intersection with $\S_{\f0}(n,\sqrt{A}-\sqrt{\theta_n})$. The packing density $\Updelta_n(\mathscr{S})$ is defined as the ratio of the saturated packing volume to the larger sphere's volume $\text{Vol}\left(\S_{\f0}(n,\sqrt{A}-\sqrt{\theta_n})\right)$, i.e.,
\begin{align}
    \Updelta_n(\mathscr{S}) \triangleq \frac{\text{Vol}\left( \S_{\f0}(n,\sqrt{A}-\sqrt{\theta_n}) \cap \bigcup_{i=1}^{M}\S_{\bar{\fc}_i}(n,r_0)\right)}{\text{Vol}\left(\S_{\f0}(n,\sqrt{A}-\sqrt{\theta_n})\right)} \,.\,
    \label{Eq.DensitySphere}
\end{align}
Sphere packing $\mathscr{S}$ is called \emph{saturated} if no spheres can be added to the arrangement without overlap.

In particular, we use a packing argument that has a similar flavor as that observed in the Minkowski--Hlawka theorem for saturated packing \cite{CHSN13}.
Specifically, consider a saturated packing arrangement of 
\begin{align}
    \label{Eq.Union_Spheres}
    \bigcup_{i=1}^{M(n,R)} \S_{\fc_i}(n,\sqrt{\theta_n})
\end{align}
spheres with radius $r_0=\sqrt{\theta_n}$ embedded within sphere $\S_{\f0}(n,\sqrt{A}-\sqrt{\theta_n})$. Then, for such an arrangement, we have the
\begin{wrapfigure}[21]{r}{0.4\textwidth}
\vspace{-\intextsep}
\vspace{-1mm}
  \begin{center}
    \scalebox{.8}{

\begin{tikzpicture}[scale=.55][thick]

\draw[dashed] (0,0) circle (4.05cm);
\draw[blau_2b, line width=1pt] (0,0) circle (3.05cm);

\draw (0,0) circle (1cm);
\draw [fill=gray!10, fill opacity=1] (0,0) circle (1cm);
\fill (0,0) circle [radius=1.5pt];

\draw (2,0) circle (1cm);
\draw [fill=gray!10, fill opacity=1] (2,0) circle (1cm);

\draw (1,1.73) circle (1cm);
\draw [fill=gray!10, fill opacity=1] (1,1.73) circle (1cm);

\draw (-1,1.73) circle (1cm);
\draw [fill=gray!10, fill opacity=1] (-1,1.73) circle (1cm);

\draw (-2,0) circle (1cm);
\draw [fill=gray!10, fill opacity=1] (-2,0) circle (1cm);

\draw (-1,-1.73) circle (1cm);
\draw [fill=gray!10, fill opacity=1] (-1,-1.73) circle (1cm);

\draw (1,-1.73) circle (1cm);
\draw [fill=gray!10, fill opacity=1] (1,-1.73) circle (1cm);


\draw (3,-1.73) circle (1cm);
\draw [fill=mycolor9!30, fill opacity=0.4] (3,-1.73) circle (1cm);
\fill (3,-1.73) circle [radius=1.5pt];

\draw (2.05,3.49) circle (1cm);
\draw [fill=orange_6b!50, fill opacity=0.4] (2.05,3.49) circle (1cm);
\fill (2.05,3.49) circle [radius=1.5pt];

\draw (-2.05,3.49) circle (1cm);
\draw [fill=orange_6b!50, fill opacity=0.4] (-2.05,3.49) circle (1cm);
\fill (-2.05,3.49) circle [radius=1.5pt];

\draw (2.05,-3.49) circle (1cm);
\draw [fill=orange_6b!50, fill opacity=0.4] (2.05,-3.49) circle (1cm);
\fill (2.05,-3.49) circle [radius=1.5pt];

\draw (-2.05,-3.49) circle (1cm);
\draw [fill=orange_6b!50, fill opacity=0.4] (-2.05,-3.49) circle (1cm);
\fill (-2.05,-3.49) circle [radius=1.5pt];

\draw (4.05,0) circle (1cm);
\draw [fill=orange_6b!50, fill opacity=0.4] (4.05,0) circle (1cm);
\fill (4.05,0) circle [radius=1.5pt];

\draw (-4.05,0) circle (1cm);
\draw [fill=orange_6b!50, fill opacity=0.4] (-4.05,0) circle (1cm);
\fill (-4.05,0) circle [radius=1.5pt]; 

\draw (3,1.73) circle (1cm);
\draw [fill=mycolor9!30, fill opacity=0.4] (3,1.73) circle (1cm);
\fill (3,1.73) circle [radius=1.5pt];

\draw (0,2*1.73) circle (1cm);
\draw [fill=mycolor9!30, fill opacity=0.4] (0,2*1.73) circle (1cm);
\fill (0,2*1.73) circle [radius=1.5pt];

\draw (-3,1.73) circle (1cm);
\draw [fill=mycolor9!30, fill opacity=0.4] (-3,1.73) circle (1cm);
\fill (-3,1.73) circle [radius=1.5pt];

\draw (-3,-1.73) circle (1cm);
\draw [fill=mycolor9!30, fill opacity=0.4] (-3,-1.73) circle (1cm);

\draw (0,-2*1.73) circle (1cm);
\draw [fill=mycolor9!30, fill opacity=0.4] (0,-2*1.73) circle (1cm);
\fill (0,-2*1.73) circle [radius=1.5pt];

\draw [blau_2b] (0,0) -- (2.85,-1.03);
\node at (4.17,-1.20) {$\sqrt{A}-\sqrt{\theta_n}$};


\draw [dashed] (0,0) -- (-2.03,3.475);
\node at (-2.5,3.9) {$\sqrt{A}$};

\draw [mycolor9] (-3,-1.73) -- (-3.8,-2.3);
\node at (-4.37,-2.67) {$\sqrt{\theta_n}$};
\fill (-3,-1.73) circle [radius=1.5pt];

\end{tikzpicture}}
  \end{center}
  \captionsetup{justification=justified}
  \vspace{-4mm}
\caption{\footnotesize Illustration of a saturated sphere packing inside a hyper sphere, where small spheres of radius $r_0 = \sqrt{\theta_n}$ cover a larger hyper sphere. The small spheres are disjoint from each other and have a non-empty intersection with the large sphere. Some of the small spheres, colored in green, are not entirely contained within the larger sphere, and yet they are considered to be a part of the packing arrangement, since their centers fulfill the power constraint in \ref{Ineq.Ave_Pow_Constr}. Yellow colored spheres whose centers exactly lies on the circle with radius $A$ do not contribute to the packing. As we assign a codeword to each sphere center (white and green), the $2$-norm of a codeword is bounded by $\sqrt{A}$ as required.}
\end{wrapfigure}
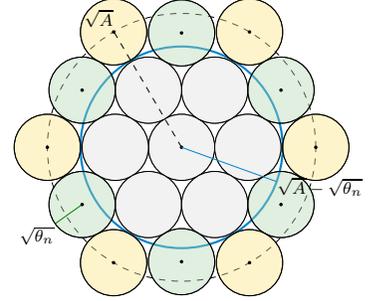
following lower \cite[Lem.~2.1]{C10} and upper bounds \cite[Eq.~45]{CHSN13} on the packing density
\begin{align}
    \label{Ineq.Density}
    2^{-n} \leq \Updelta_n(\mathscr{S}) \leq 2^{-0.599n} \;.\,
\end{align}

In our subsequent analysis, we use the above lower bound which can be proved as follows: For the saturated packing arrangement given in \eqref{Eq.Union_Spheres}, there cannot be a point in the larger sphere $\S_{\f0}(n,\sqrt{A}-\sqrt{\theta_n})$ with a distance of more than $2r_0$ from all sphere centers. Otherwise, a new sphere could be added which contradicts the assumption that the union of $M(n,R)$ spheres with radius $\sqrt{\theta_n}$ is saturated. Now, if we double the radius of each sphere, the spheres with radius $2r_0$ cover thoroughly the entire volume of $\S_{\f0}(n,\sqrt{A}-\sqrt{\theta_n})$, that is, each point inside the large hyper sphere $\S_{\f0}(n,\sqrt{A}-\sqrt{\theta_n})$ belongs to at least one of the small spheres. In general, the volume of a hyper sphere of radius $r$ is given by \cite[Eq.~(16)]{CHSN13}
\begin{align}
    \text{Vol}\left(\S_{\fx}(n,r)\right) = \frac{\pi^{\frac{n}{2}}}{\Gamma(\frac{n}{2}+1)} \cdot r^{n} \,.\,
    \label{Eq.VolS}
\end{align}
Hence, if the radius of the small spheres is doubled, the volume of $\bigcup_{i=1}^{M(n,R)} \S_{\fc_i}(n,\sqrt{\theta_n})$ is increased by $2^n$. Since the spheres with radius $2r_0$ cover $\S_{\f0}(n,\sqrt{A}-\sqrt{\theta_n})$, it follows that the original $r_0$-radius packing has a density of at least $2^{-n}$~\footnote{\,\textcolor{mycolor5}{We note that the proposed proof of the lower bound in \eqref{Ineq.Density} is non-constructive in the sense that, while the existence of the respective saturated packing is proved, no systematic construction method is provided.}}.
We assign a codeword to the center $\fc_i$ of each small sphere. The codewords satisfy the input constraint as
\begin{align}
    \label{Ineq.Norm_Infinity}
    \norm{\bar{\fc}_i} \leq \sqrt{A} \;.\,
\end{align}

Since the volume of each sphere is equal to $\text{Vol}(\S_{\fc_1}(n,r_0))$ and the centers of all spheres lie inside the sphere, the total number of spheres is bounded from below by
\begin{align}
    \label{Eq.L_Achiev}
    M & = \frac{\text{Vol}\left(\bigcup_{i=1}^{M}\S_{\bar{\fc}_i}(n,r_0)\right)}{\text{Vol}(\S_{\fc_1}(n,r_0))}
    \nonumber\\&
    \geq \frac{\text{Vol}\left(\S_{\f0}(n,\sqrt{A}-\sqrt{\theta_n})\cap\bigcup_{i=1}^{M}\S_{\bar{\fc}_i}(n,r_0)\right)}{\text{Vol}(\S_{\bar{\fc}_1}(n,r_0))}
    \nonumber\\
    & = \frac{\Updelta_n(\mathscr{S}) \cdot
    \text{Vol}\left(\S_{\f0}(n,\sqrt{A}-\sqrt{\theta_n})\right)}{\text{Vol}(\S_{\bar{\fc}_1}(n,r_0))}
    \nonumber\\&
    \geq 2^{-n} \cdot \frac{\text{Vol}\left(\S_{\f0}(n,\sqrt{A}-\sqrt{\theta_n})\right)}{\text{Vol}(\S_{\bar{\fc}_1}(n,r_0))}
    \,,\,
\end{align}
where the first inequality holds by (\ref{Eq.DensitySphere}) and the second inequality holds by (\ref{Ineq.Density}).
The above bound can be further simplified as follows
\begin{align}
    \log M & \stackrel{(a)}{\geq} \log \left( \frac{\sqrt{A}-\sqrt{\theta_n}}{r_0} \right)^n - n
    \nonumber\\&
    \stackrel{(b)}{=} n \log \left( \frac{\sqrt{A}-\sqrt{\theta_n}}{\sqrt{\theta_n}} \right) - n
    \nonumber\\& 
    = n \log \left( \sqrt{\frac{A}{\theta_n}} - 1 \right) - n
    \nonumber\\&
    \stackrel{(c)}{\geq} \frac{1}{2} n \log \left( \frac{A}{\theta_n} \right) - 2n \;,\,
\end{align}
where $(a)$ exploits (\ref{Eq.VolS}), $(b)$ follows from $r_0 = \sqrt{\theta_n}$, and $(c)$ holds by $\log (t-1) \geq \log t - 1 \,, \forall t \geq 2$.
Therefore, for $\theta_n = A / n^{\frac{1}{2}(1-(b+\kappa))}$, we obtain
\begin{align}
    \log M & \geq \frac{1}{2} n \log n^{\frac{1}{2}(1-(b+\kappa))} - 2n
    \nonumber\\
    & = \left( \frac{1-\left(b+\kappa\right)}{4} \right) n\log n - 2n \,,
    \label{Eq.Log_L}
\end{align}
where the dominant term is of order $n \log n$. Hence, for obtaining a finite value for the lower bound of the rate, $R$, \eqref{Eq.Log_L} induces the scaling law of $M$ to be $2^{(n\log n)R}$. Therefore, we obtain
\begin{align}
    R & \geq \frac{1}{\log n} \left[ \left( \frac{1-\left(b+\kappa\right)}{4} \right) \log n - 2 \right] \;,\,
\end{align}
which tends to $\frac{1-\kappa}{4}$ when $n \to \infty$ and $b\rightarrow 0$.
\subsubsection*{Encoding}
Given message $i\in [\![M]\!]$, transmit $\bar{\fx} = \Bar{\fc}_i$.
\subsubsection*{Decoding}
Let
\begin{align}
    \tau_n = \frac{\gamma^2\theta_n}{3} = \frac{A\gamma^2}{3n^{\frac{1}{2}(1-(b+\kappa))}} \;,\,
    \label{Eq.tau_n}
\end{align}
where $0 < b < 1$ is an arbitrarily small constant, $0<c<2$ is a constant, $\kappa \in [0,1)$, and $\gamma$ is the infimum value of all fading coefficients $g$.

To identify whether message $j\in \M$ was sent, given the fading coefficient $g$, the decoder checks whether the channel output $\bar{\fy}$ belongs to the following decoding set:
\begin{align}
    \sT_{\mathbbmss{K}} = \bigcup_{j \in \mathbbmss{K}} \mathbbmss{T}_{j,g} \,,
\end{align}
where
\begin{align}
    \mathbbmss{T}_{j,g} = \left\{ \bar{\fy} \in \mathbb{R}^n \,:\; \sum_{t=1}^n (\bar{y}_t - g \bar{c}_{j,t})^2
    \leq \sigma_Z^2 + \tau_n \right\} \;.\,
    \label{Eq.DecodingTerritory}
\end{align}
is referred to as the individual decoding territory evaluated for observation vector $\fy$ and codeword $\fc_j$.
\subsubsection*{Error Analysis}
Fix $e_1,e_2 > 0$ and let $\zeta_0, \zeta_1 > 0$ be arbitrarily small constants. Before we proceed, for the sake of brevity of analysis, we introduce the following conventions:
\begin{itemize}
    \item Let $Y_t(.|i,g)$ denote the channel output at time $t$ given that $\Bar{\fx} = \Bar{\fc}_i$ and $G = g$.
    \item $\fY(.|i,g) = (Y_1(.|i,g),\ldots,Y_n(.|i,g))$.
\end{itemize}

Consider the type I errors, i.e., the transmitter sends $\Bar{\fc}_i$, yet $\fY(.|i,g) \notin \mathbbmss{T}_{\mathbbmss{K},g}$. For every $i \in [\![M]\!]$, the type I error probability is given by
\begin{align}
    P_{e,1}(i) & = \sup_{g\in\G} \left[ P_{e,1} \left( i \,| g \; \right) \right] \;,\,
    \label{Eq.TypeIError-Sup}
\end{align}
where
\begin{align}
    \label{Eq.TypeIError-Sup-1}
    P_{e,1} \left( i \,| g \; \right) & = \Pr \left( \bar{\fY}(.|i,g) \in \mathbbmss{T}_{\mathbbmss{K},g}^c \right)
    \nonumber\\
    & = \Pr \left( \bar{\fY}(.|i,g) \in \left( \bigcup_{i \in \mathbbmss{K}} \mathbbmss{T}_{i,g} \right)^c \right)
    \nonumber\\
    & \stackrel{(a)}{=} \Pr \left( \bar{\fY}(.|i,g) \in \bigcap_{i \in \mathbbmss{K}} \mathbbmss{T}_{i,g}^c \right)
    \nonumber\\
    & \stackrel{(b)}{\leq} \Pr \left( \bar{\fY}(.|i,g) \in \mathbbmss{T}_{i,g}^c \right)
    \nonumber\\
    & \stackrel{(c)}{\equiv} \Pr \left( \sum_{t=1}^n (\bar{Y}_t(.|i,g) - G\bar{c}_{i,t})^2 > \sigma_Z^2+\tau_n \right)
    \nonumber\\
    & \stackrel{(d)}{=} \Pr \left(\sum_{t=1}^n {\bar{Z}_t}^2 > \sigma_Z^2 + \tau_n \right) \;,\,
\end{align}
where $(a)$ follows by \emph{De Morgan}'s law for finite number of unions, i.e., $\left( \bigcup_{i \in \mathbbmss{K}} \mathbbmss{T}_{i,g} \right)^c = \bigcap_{i \in \mathbbmss{K}} \mathbbmss{T}_{i,g}^c$, $(b)$ holds since $\bigcap_{i \in \mathbbmss{K}} \mathbbmss{T}_{i,g}^c \subset \mathbbmss{T}_{i,g}$, $(c)$ follows by definition of the individual decoding territory in \eqref{Eq.DecodingTerritory}, and $(d)$ holds since the fading coefficient $G$ and the noise vector $\bar{\fZ}$ are statistically independent.

Now, in order to bound $P_{e,1} \left( i \,| g \; \right)$, we apply Chebyshev's inequality, namely
\begin{align}
    P_{e,1} \left( i \,| g \right) & \leq \Pr\left( \sum_{t=1}^n {\bar{Z}_t}^2 - \sigma_Z^2 > \tau_n \right)
    \nonumber\\
    & \stackrel{(a)}{\leq} \frac{3\sigma_Z^4}{n\tau_n^2}
    \nonumber\\
    & \stackrel{(b)}{=} \frac{27\sigma_Z^4}{A^2 \gamma^4 n^{\kappa+b}}
    \nonumber\\
    & \leq e_1 \;,\,
\end{align}
where $(a)$ holds since the fourth moment of a Gaussian variable $V\sim \N(0,\sigma_V^2)$ is $\mathbb{E}[V^4]=3\sigma_V^4$ and $(b)$ follows from \eqref{Eq.tau_n}. Hence, $P_{e,1} \left( i \, | g \,\right) \leq e_1 \;,\, \forall g\in\G$ holds for sufficiently large $n$ and arbitrarily small $e_1 > 0$. Thereby, the type I error probability satisfies $P_{e,1}\left( i \right) \leq e_1$; see \eqref{Eq.TypeIError-Sup}.

Next, we address type II errors, i.e., when $\Bar{\fY}(.|i,g) \in \mathbbmss{T}_{\mathbbmss{K},g}$ while the transmitter sent $\Bar{\fc}_i$ with $i \notin \mathbbmss{K}$.
Then, for every $\mathbbmss{K} \in \binom{M}{K}$, where $i\notin \mathbbmss{K}$, the type II error probability is given by
\begin{align}
    \label{Eq.TypeIIError-Sup}
    P_{e,2}(i,\mathbbmss{K})  = \sup_{g\in\G} \left[ P_{e,2} \left( i,\mathbbmss{K} \, | g \; \right) \right] \;,\,
\end{align}
where
\begin{align}
    \label{Eq.TypeIIError-Sup-1}
    P_{e,2} \left( i,\mathbbmss{K} \,| g \right) & = \Pr \left( \bar{\fY}(.|i,g) \in \mathbbmss{T}_{\mathbbmss{K},g} \right)
    \nonumber\\&
    = \Pr \left( \bar{\fY}(.|i,g) \in \left( \bigcup_{j \in \mathbbmss{K}} \mathbbmss{T}_{j,g} \right) \right)
    \nonumber\\
    & \equiv \Pr\left( \bigcup_{j \in \mathbbmss{K}} \left\{
    \sum_{t=1}^n \left(\bar{Y}_t(.|i,g) - G\bar{c}_{j,t}\right)^2
    \leq \sigma_Z^2+\tau_n \right\} \right)
    \nonumber\\
    & \stackrel{(a)}{=} \Pr\left(\bigcup_{j \in \mathbbmss{K}} \left\{ \sum_{t=1}^n \left(g\left(\bar{c}_{i,t}-\bar{c}_{j,t}\right)+\bar{Z}_t\right)^2
    \leq \sigma_Z^2+\tau_n \right\} \right)
    \nonumber\\
    & \stackrel{(b)}{\leq} \sum_{j \in \mathbbmss{K}} \Pr \left( \sum_{t=1}^n \left(g\left(\bar{c}_{i,t}-\bar{c}_{j,t}\right)+\bar{Z}_t\right)^2
    \leq \sigma_Z^2+\tau_n \right) \,,\,
\end{align}
where $(a)$ hold since the fading coefficient $G$ and the noise vector $\bar{\fZ}$ are statistically independent and $(b)$ follows by the union bound, i.e., the probability of union of events is upper bounded by sum of probability of the individual events.

In order to bound \eqref{Eq.TypeIIError-Sup-1}, we divide into two cases. First, consider $g \in \G$ such that $\| g ( \bar{\fc}_i - \bar{\fc}_j ) \| > 2 \sqrt{\sigma_Z^2 + \tau_n}$.
Therefore, by the reverse triangle inequality, $\norm{\fa - \fb} \geq \left| \norm{\fa} - \norm{\fb} \right|$, we have
\begin{align}
    \sqrt{ \sum_{t=1}^n \left( \left( g\left( \bar{c}_{i,t} - \bar{c}_{j,t} \right) \right) + \bar{Z}_t \right)^2} & \geq \norm{g \left( \bar{\fc}_i - \bar{\fc}_j \right)} - \norm{\bar{\fZ}}
    \nonumber\\
    & \geq 2 \sqrt{\sigma_Z^2 + \tau_n} - \norm{\bar{\fZ}} \,.\,
\end{align}
Hence, for every $g$ such that $\| g \left( \bar{\fc_i} - \bar{\fc_j} \right) \| > 2 \sqrt{\sigma_Z^2 + \tau_n}$, we can bound the type II error probability by
\begin{align}
    \label{Ineq.TypeII_Slow_Conditional1}
    P_{e,2} \left( i,\mathbbmss{K} \,\big| g \,\right) & \leq \sum_{j \in \mathbbmss{K}} \Pr\left( \norm{\bar{\fZ}} \geq \sqrt{\sigma_Z^2 + \tau_n} \right)
    \nonumber\\
    & = \sum_{j \in \mathbbmss{K}} \Pr\left(\sum_{t=1}^n {\bar{Z}_t}^2 > \sigma_Z^2 + \tau_n \right)
    \nonumber\\
    & \leq \frac{3K\sigma_Z^4}{n\tau_n^2} 
    \nonumber\\
    & = \frac{27\sigma_Z^4}{A^2 \gamma^4 n^b}
    \nonumber\\
    & \leq e_2 \;,\,
\end{align}
where $(a)$ follows from applying Chebyshev's inequality and since the fourth moment of a Gaussian variable $V\sim \N(0,\sigma_V^2)$ is $\mathbb{E}[V^4] = 3\sigma_V^4$ and $(b)$ follows from \eqref{Eq.tau_n}. Hence, $P_{e,1} \left( i \, | g \,\right) \leq e_1 \;,\, \forall g\in\G$ holds for sufficiently large $n$ and arbitrarily small $e_1 > 0$. Thereby, the type I error probability satisfies $P_{e,2}\left( i,\mathbbmss{K} \right) \leq e_2$; see \eqref{Eq.TypeIError-Sup}.

Now, we focus on the second case, i.e., when
\begin{align}
    \label{Ineq.Bound_G}
    \norm{g \left( \bar{\fc}_i - \bar{\fc}_j \right)} \leq 2 \sqrt{\sigma_Z^2 + \tau_n} \;.\,
\end{align}
Observe that for every given $g\in\G$,
\begin{align}
    \sum_{t=1}^n (g(\bar{c}_{i,t}-\bar{c}_{j,t}) + \bar{Z}_t)^2 = \sum_{t=1}^n g^2(\bar{c}_{i,t}-\bar{c}_{j,t})^2 + \sum_{t=1}^n \bar{Z}_t^2 + 2\sum_{t=1}^n g(\bar{c}_{i,t} - \bar{c}_{j,t})Z_t \;.\,
     \label{Eq.Pe2normSlow}
\end{align}
Then define the event
\begin{align}
    \E_0(\fZ|g) = \left\{ \fZ \in \mathbb{R}^n \;:\, \Big| \sum_{t=1}^n g(\bar{c}_{i,t}-\bar{c}_{j,t})\bar{Z}_t \Big| > \frac{\tau_n}{2} \right\} \;,\,
    \label{Eq.E0FadingSlow}
\end{align}
Now, in order to bound $\Pr(\E_0(\fZ|g))$, we apply Chebyshev's inequality, namely
\begin{align}
    \Pr\left( \E_0(\fZ|g) \right) & \leq \frac{\text{Var}\left[ \sum_{t=1}^n g(\bar{c}_{i,t}-\bar{c}_{j,t}) \bar{Z}_t \right]}{\left( \tau_n / 2 \right)^2}
    \nonumber\\&
    \stackrel{(a)}{=} \frac{4 \sum_{t=1}^n g^2(\bar{c}_{i,t}-\bar{c}_{j,t})^2 \mathbb{E}[\bar{Z}_t^2]}{\tau_n^2}
    \nonumber\\&
    \stackrel{(b)}{=} \frac{4\sigma_Z^2 \| g (\bar{\fc}_i-\bar{\fc}_j) \|^2}{n\tau_n^2}
    \nonumber\\&
    \stackrel{(c)}{=} \frac{16 \sigma_Z^2 \left( \sigma_Z^2 + \tau_n \right)}{n\tau_n^2}
    \nonumber\\&
    = \frac{144\sigma^2_Z\left( \sigma_Z^2 + \tau_n \right)}{A^2\gamma^4n^{\kappa+b}}
    \nonumber\\&
    \overset{\text{\scriptsize def}}{=} \zeta_0 \;,\,
    \label{Eq.E0-Event}
\end{align}
where $(a)$ and $(b)$ holds since the noise sequence $\Bar{\fZ} \overset{\text{\scriptsize i.i.d.}}{\sim} \mathcal{N}\big(0,\frac{\sigma_Z^2}{n}\big)$, that is, $\text{Var}[\Bar{Z}_t] = \mathbb{E}[\bar{Z}_t^2] - \mathbb{E}^2[\bar{Z}_t] = \frac{\sigma_Z^2}{n}$, and $(c)$ follows from \eqref{Ineq.Bound_G}. Observe that given the complementary event $\E_0^c(\fZ|g)$, we have
\begin{align}
    2\sum_{t=1}^n g\left( \bar{c}_{i,t} - \bar{c}_{j,t} \right) \bar{Z}_t \geq - \tau_n \;,\,
\end{align}
Therefore, the event $\E_0^c(\fZ|g)$, the type II error event in \eqref{Eq.TypeIIError-Sup-1}, and the identity in \eqref{Ineq.Bound_G} together imply that the following event occurs,
\begin{align}
    \E_1(\fZ|g) = \left\{ \fZ \in \mathbb{R}^n \;:\, \sum_{t=1}^n g^2 (\bar{c}_{i,t} - \bar{c}_{j,t})^2 + \sum_{t=1}^n \bar{Z}_t^2 \leq \sigma_Z^2 + 2\tau_n \right\} \;.\,
    \label{Eq.Pe2normConsequenceSlow}
\end{align}
Now lets define
\begin{align}
    \mathcal{H}_{i,j}(\fZ|g) = \left\{ \fZ \in \mathbb{R}^n \;:\, \sum_{t=1}^n (g (\bar{c}_{i,t} - \bar{c}_{j,t}) + \bar{Z}_t)^2 \leq \sigma_Z^2 + \tau_n \right\} \;.\,
\end{align}
Therefore, applying the law of total probability to (\ref{Eq.TypeIIError-E_0+E_1}), we have
\begin{align}
    P_{e,2} ( i,\mathbbmss{K} | \,g ) & = \sum_{j \in \mathbbmss{K}} \left[
    \Pr\left( \mathcal{H}_{i,j}(\fZ|g) \cap \E_0(\fZ|g) \right) + \Pr\left( \mathcal{H}_{i,j}(\fZ|g) \cap {\E_0^c(\fZ|g)} \right) \right]
    \nonumber\\
    & \leq \sum_{j \in \mathbbmss{K}} \left[ \Pr(\E_0(\fZ|g)) + \Pr\left( \E_1(\fZ|g) \right) \right]
    \nonumber\\
    & \leq K \left[ \zeta_0 + \Pr\left( \E_1(\fZ|g) \right) \right] \;,\,
    \label{Eq.TypeIIError-E_0+E_1}
\end{align}
where the last inequality holds by (\ref{Eq.E0-Event}).

We now proceed with bounding $\Pr(\E_1(\fZ|g))$ as follows. Based on the codebook construction, each codeword is surrounded by a sphere of radius $\sqrt{\theta_n}$, that is 
\begin{align}
    \norm{\bar{\fc}_i - \bar{\fc}_j} \geq \sqrt{\theta_n} \;.\,
\end{align}
which implies
\begin{align}
    g^2 \norm{\bar{\fc}_i - \bar{\fc}_j}^2 \geq  \gamma^2 \theta_n \;,\,
\end{align}
where $\gamma$ is the infimum value in ${\G}$. Thus, we can establish the following upper bound for event $\E_1(\fZ|g)$:
\begin{align}
    \label{Ineq.E_1}
    \Pr(\E_1(\fZ|g)) & \leq \Pr\left( \norm{\bar{\fZ}}^2 \leq \sigma_Z^2 + 2 \tau_n - \gamma^2 \theta_n  \right)
    \nonumber\\&
    = \Pr\left( \norm{\bar{\fZ}}^2 - \sigma_Z^2 \leq - \tau_n \right)
    \nonumber\\&
    = \Pr\left(\sum_{t=1}^n \bar{Z}_t^2 - \sigma_Z^2 \leq - \tau_n \right)
    \nonumber\\&
    \stackrel{(a)}{\leq} \frac{\sum_{t=1}^n\text{Var}[\bar{Z}_t^2]}{ \tau_n^2}
    \nonumber\\
    & \stackrel{(b)}{\leq} \frac{\sum_{t=1}^n \mathbb{E}[ \bar{Z}_t^4 ]}{ \tau_n^2} 
    \nonumber\\
    & = \frac{ 3n \left(\frac{\sigma_Z^2}{n}\right)^2}{ \tau_n^2}
    \nonumber\\
    & = \frac{3\sigma_Z^4}{n \tau_n^2}
    \nonumber\\
    & \stackrel{(c)}{=} \frac{27\sigma_Z^4}{A^2 \gamma^4 n^{\kappa+b}}
    \nonumber\\&
    \overset{\text{\scriptsize def}}{=} \zeta_1 \;,\,
\end{align}
where $(a)$ follows from applying Chebyshev's inequality, $(b)$ holds since the fourth moment of a Gaussian variable $V\sim \N(0,\sigma_V^2)$ is $\mathbb{E}[V^4] = 3\sigma_V^4$ and $(c)$ follows from \eqref{Eq.tau_n} and \eqref{Eq.E0-Event}. Therefore, we can proceed to bound the rightmost in \eqref{Eq.TypeIIError-E_0+E_1} as follows
\begin{align}
    \label{Ineq.TypeII_Slow_Conditional2}
    P_{e,2}(i,\mathbbmss{K}) & \leq K \left[ \Pr(\E_0(\fZ|g)) + \Pr(\E_1(\fZ|g)) \right]
    \nonumber\\&
    \leq K \left[ \zeta_0 + \zeta_1 \right]
    \nonumber\\&
    = \frac{144K\sigma^2_Z\left( \sigma_Z^2 + \tau_n \right)}{A^2\gamma^4n^{\kappa+b}} + \frac{27K\sigma_Z^4}{A^2 \gamma^4 n^{\kappa+b}}
    \nonumber\\&
    = \frac{144\sigma^2_Z\left( \sigma_Z^2 + \tau_n \right) + 27\sigma_Z^4}{A^2\gamma^4n^b}
    \nonumber\\&
    \leq e_2 \,,\,
\end{align}
hence, $P_{e,2} \left( i,\mathbbmss{K} \, | g \,\right) \leq e_2 \;,\, \forall g\in\G$ holds for sufficiently large $n$ and arbitrarily small $e_2 > 0$.
Thereby, the type II error probability satisfies $P_{e,2}\left( i,\mathbbmss{K} \right) \leq e_2$; see \eqref{Eq.TypeIIError-Sup}.

We have thus shown that for every $e_1,e_2>0$ and sufficiently large $n$, there exists an $(n, M(n,R), K(n,\kappa), \allowbreak e_1, e_2)$ code.
\subsection{Converse Proof}
\label{App.Conv}
The converse proof consists of the following two main steps.
\begin{itemize}
    \item \textbf{Step 1:} We show in Lemma~\ref{Lem.DConv} that for any achievable rate (for which the type I and type II error probabilities vanish as $n\to\infty$), the distance between every pair of codeword should be at least larger than a threshold.
    \item \textbf{Step~2:} Employing the Lemma~\ref{Lem.DConv}, we derive an upper bound on the codebook size of achievable DKI codes.
\end{itemize}
We start with the following lemma which establish a lower bound on the Euclidean norm of two different codewords' difference.
\begin{lemma}
\label{Lem.DConv}
Suppose that $R$ is an achievable rate for the GSF $\bG$ and let $b>0$ be an arbitrarily small constant that does not depend on codeword length $n$. Consider a sequence of $(n, M(n,R), K(n,\kappa), \allowbreak e_1^{(n)}, \allowbreak e_2^{(n)})$ codes $(\C^{(n)},\T^{(n)})$ such that $e_1^{(n)}$ and $e_2^{(n)}$ tend to zero as $n\rightarrow\infty$. Then there exists $n_0(b)$, such that for all $n>n_0(b)$, every pair of codewords in the codebook $\C^{(n)}$ satisfies the following property.

For every pair of codewords, $\fc_{i_1}$ and $\fc_{i_2}$,
\begin{subequations}
\begin{align}
    \label{Ineq.Converse_Lem}
    \norm{\fc_{i_1} - \fc_{i_2}}\geq 2\sqrt{n\epsilon_n'} \;,\,
\end{align}
for all $i_1,i_2\in [\![M]\!]$, such that $i_1\neq i_2$, with
\begin{align}
    \epsilon_n' = \frac{A}{n^{2 (1 + \kappa+b)}} \;,\,
\end{align}
\end{subequations}
\end{lemma}
\begin{proof}
The proof is given in Appendix~\ref{App.Converse}.
\end{proof}
Next, we use Lemma~\ref{Lem.DConv} to prove the upper bound on the DKI capacity. Observe that Lemma~\ref{Lem.DConv} implies that the distance between every pair of codewords satisfies
\begin{align}
   \norm{\fc_{i_1} - \fc_{i_2}} \geq 2\sqrt{n\epsilon_n'} \;.\,
\end{align}
Thus, we can define an arrangement of non-overlapping spheres $\S_{\fc_i}(n,\sqrt{n\epsilon_n'})$, i.e., spheres of radius $\sqrt{n\epsilon_n'}$ that are centered at the codewords $\fc_i$. Since the codewords all belong to a large hyper sphere $\S_{\f0}(n,\sqrt{nA})$ of radius $\sqrt{nA}$, it follows that the number of packed small spheres, i.e., the number of codewords $M$, is bounded by
\begin{align}
    \label{Eq.L}
    M & = \frac{\text{Vol}\left(\bigcup_{i=1}^{M} \S_{\fc_i}(n,r_0)\right)}{\text{Vol}(\S_{\fc_1}(n,\sqrt{nA}+r_0))}
    \nonumber\\
    & \stackrel{(a)}{=} \Updelta_n(\mathscr{S}) \cdot \frac{\text{Vol}\left(\S_{\f0}(n,\sqrt{nA} + r_0)\right)}{\text{Vol}(\S_{\fc_1}(n,r_0))}
    \nonumber\\
    & \stackrel{(b)}{\leq} 2^{-0.599n} \cdot \frac{\text{Vol}\left(\S_{\f0}(n,\sqrt{nA} + r_0)\right)}{\text{Vol}(\S_{\fc_1}(n,r_0))} \;,\,
\end{align}
where $(a)$ holds by definition of packing density, $(b)$ follows from inequality (\ref{Ineq.Density}). The above bound can be further simplified as
follows
\begin{align}
    \label{Eq.Converse_Log_L}
    \log M & \stackrel{(a)}{\leq} \log \left( \frac{\sqrt{nA} + r_0}{r_0} \right)^n - 0.599n
    \nonumber\\&
    \leq n\log \left( \frac{\sqrt{nA} + r_0}{r_0} \right) - 0.599n
    \nonumber\\&
    \stackrel{(b)}{=} \frac{1}{2} n \log \left( \frac{A}{\epsilon_n'} + 1 \right) - 0.599n \,,
\end{align}
where $(a)$ exploits \eqref{Eq.VolS} and $(b)$ follows from $r_0 = \frac{1}{2} ( 2\sqrt{n\epsilon_n'} )$. Therefore, for $\epsilon_n' = A / n^{2(1+\kappa+b)}$, we obtain
\begin{align}
    \log M & \leq \frac{1}{2} n \log \left( n^{2(1+\kappa+b)} + 1 \right) - 0.599n
    \nonumber\\
    & = \frac{1}{2} n \log \left( n^{2(1+\kappa+b)} \left( 1 + 1 / n^{2(1+\kappa+b)} \right) \right) - 0.599n
    \nonumber\\
    & = \frac{1}{2} n \log \left( n^{2(1+\kappa+b)} \right) + \frac{1}{2} n \log \left( 1 + 1 / n^{2(1+\kappa+b)} \right) - 0.599n
    \nonumber\\
    & = (1+\kappa+b) \, n\log n + \frac{1}{2} n \log \left( 1 + 1 / n^{2(1+\kappa+b)} \right) - 0.599n \,,
\end{align}
where the dominant term is again of order $n \log n$. Hence, for obtaining a finite value for the upper bound of the rate, $R$, \eqref{Eq.Converse_Log_L} induces the scaling law of $M$ to be $2^{(n\log n)R}$. Hence, we obtain
\begin{align}
    R & \leq \frac{1}{n\log n} \left[ (1+\kappa+b) \, n\log n + \frac{1}{2} n \log \left( 1 + 1 / n^{2(1+\kappa+b)} \right) - 0.599n \right]
    \nonumber\\
    & = 1 + \kappa + b + \log \left( 1 + 1 / n^{2(1+\kappa+b)} \right) / \log n - 0.599 / \log n
    \;,\,
\end{align}
which tends to $1 + \kappa$ as $n \to \infty$ and $b \to 0$.  This completes the proof of Theorem~\ref{Th.DKI-Capacity}.
\section{Summary and Future Directions}
\label{Sec.Summary}
In this work, we studied the DKI problem over the GSF with $K$ number target messages. We assumed that $K=K(n,\kappa)=2^{\kappa\log n} = n^{\kappa}$ where $\kappa \in [0,1)$ scales sub-linearly with the codeword length $n$. In practice, the receiver sometimes suspend the exact matching task as is considered for the standard identification \cite{Salariseddigh_IT,Salariseddigh_IT} and requires only to spot an object among a group, therefore, our results in this paper may serve as a model for event-triggered based tasks in the context of many practical XG applications where population of the target group scales sub-linearly in the codeword length. Especially, we obtained lower and upper bounds on the DKI capacity of the GSF with $K=2^{\kappa\log n}$ many target messages subject to average power constraint with the codebook size of $M(n,R)=2^{(n\log n)R}=n^{nR}$. Our results for the DKI capacity of the GSF revealed that the super-exponential scale of $n^{nR}=2^{(n\log n)R}$ is again the appropriate scale for codebook size. This scale coincides as of the codebook for the memoryless GSF and Gaussian channels \cite{Salariseddigh_IT,Salariseddigh_ITW} and stands considerably different from the traditional scales in transmission and RI setups where corresponding codebooks size grows exponentially and double exponentially, respectively.

We show the achievability proof using a packing of hyper spheres and a distance decoder. In particular, we pack hyper spheres with radius $\sqrt{n\theta_n} \sim n^{\frac{1+\kappa}{4}}$ where $\kappa \in [0,1)$ is the target identification rate, inside a larger hyper sphere, which results in $\sim 2^{((1-\kappa)/4) n\log n}$ codewords.
For the converse proof, we follow a similar approach as in our previous work for the standard identification over the slow fading channel \cite{Salariseddigh_IT,Salariseddigh_arXiv_ITW}. In general, the derivation here is more involved than the derivation in the standard identification case \cite{Salariseddigh_ITW} and entails employing of new analysis and inequalities. In our previous work on Gaussian channels with slow fading \cite{Salariseddigh_ITW}, the converse proof was based on establishing a minimum distance between each pair of codewords. Here, we incorporate effect of the number of target messages into the minimum distance in the relevant Lemma; see Eq.~1 \ref{Lem.DConv}.

The results presented in this paper can be extended in several directions, some of which are listed in the following as potential topics for future research works:
\begin{itemize}
    \item \textcolor{blau_2b}{\textbf{Memory}}: Including inter-symbol (ISI) interference into the channel model assuming that the degree of ISI is either constant or growing function in codeword length as observed in a recent work for Poisson channel \cite{Salariseddigh22_2}.
    \item \textcolor{blau_2b}{\textbf{Fast Fading}}: The results in this paper can be extended to the Gaussian channels with fast fading model.
    \item \textcolor{blau_2b}{\textbf{Maximum Power Constraint}}: Our achievability proof in this work consider only the average power constraint, however, an interesting future research may include both the average and maximum power constraints at the same time which seems more practical.
    \item \textcolor{blau_2b}{\textbf{Continuous Alphabet Conjecture}}: Our observations for the codebook size of following studies
        \begin{itemize}
            \item Standard identification over the Gaussian channels without memory \cite{Salariseddigh_ITW,Salariseddigh_arXiv_ITW},
            \item Standard identification over the Poisson channels without memory \cite{Salariseddigh_GC_IEEE,Salariseddigh_GC_arXiv,Salariseddigh22},
            \item Standard identification over the Poisson channels with memory \cite{Salariseddigh22_2},
            \item $K$-identification over the Slow fading channel without memory (current paper),
        \end{itemize}
    lead us to conjecture that the codebook size for every \emph{continuous} alphabet channel either in standard or $K$-identification and with/out memory is a super-exponential function, i.e., $2^{(n\log n)R}$. However, a formal proof of this conjecture remains unknown.
    \item \textcolor{blau_2b}{\textbf{Fekete's Lemma}}: Investigation of the behavior of the DKI capacity in the sense of Fekete's Lemma \cite{Boche20}: To verify whether the pessimistic ($\underline{C} = \liminf_{n \to \infty} \allowbreak \frac{\log M(n,R)}{n \log n}$) and optimistic ($\overline{C} = \limsup_{n \to \infty} \frac{\log M(n,R)}{n \log n}$) capacities \cite{A06} coincide or not; see \cite{Boche20} for more details.
    \item \textcolor{blau_2b}{\textbf{Channel Reliability Function}}: A complete characterization of the asymptotic behavior of the decoding errors as a function of the codeword length for $0 < R < C$ requires knowledge of the corresponding channel reliability function (CRF) \cite{Boche21}. To the best of the authors’ knowledge, the CRF for DKI has not been studied in the literature so far, neither for the Gaussian channel \cite{Salariseddigh_IT} nor the Poisson channel \cite{Salariseddigh_GC_IEEE,Salariseddigh_GC_arXiv,Salariseddigh22}.
    \item \textcolor{blau_2b}{\textbf{Explicit Code Construction}}: Explicit construction of DKI codes with incorporating the effect of number of target messages and the development of low-complexity encoding/decoding schemes for practical designs where the associated efficiency of such codes can be evaluated with regard to to the our derived performance bounds in Section~\ref{Sec.Res}.
    \item \textcolor{blau_2b}{\textbf{Multi User}}: The extension of this study (point-to-point system) to multi-user scenarios (e.g., broadcast and multiple access channels) or multiple-input multiple-output channels may seems more relevant in applications of complex MC nano-networks within the future generation wireless networks (XG).
\end{itemize}
\appendix
\section{Proof of Lemma~\ref{Lem.DConv}}
\label{App.Converse}
In the following, we provide the proof of Lemma~\ref{Lem.DConv}. The method of proof is by contradiction, namely, we assume that the condition given in \eqref{Ineq.Converse_Lem} is violated and then we show that this leads to a contradiction, namely, sum of the type I and type II error probabilities converge to one, i.e., $\lim_{n\to\infty} \left[ P_{e,1}(i_1) + P_{e,2}(i_2,\mathbbmss{K}) \right] = 1$.
Fix $e_1$ and $e_2$. Let $\zeta,\eta,\mu,\pi> 0$ be arbitrarily small constants. Assume to the contrary that there exist two messages $i_1$ and $i_2$, where $i_1\neq i_2$, such that
\begin{align}
    \norm{\fc_{i_1} - \fc_{i_2}} < 2\sqrt{n\epsilon_n'} = \alpha_n \;,\,
    \label{Ineq.Negated_Assump}
\end{align}
where
\begin{align}
    \label{Eq.Alpha_nSlow}
    \alpha_n \equiv \frac{2\sqrt{A}}{n^{\frac{1}{2}(1+2(\kappa+b)}} \;.\,
\end{align}
Now let us define the following subsets
\begin{align}
    \mathbbmss{P}_{i_1,i_2} = \left\{ \fy \in \mathbbmss{T}_{i_1,\fg} \,:\; \norm{\fy - g\fc_{i_2}} \leq \sqrt{n\left( \sigma_Z^2 + \zeta \right)} \right\} \;,\,
    \label{Eq.Event_P}
\end{align}
\begin{align}
    \mathbbmss{Q}_{i_1,i_2} = \left\{ \fy \in \mathbb{Y}^n \,:\; \norm{\fy - g\fc_{i_2}} \leq \sqrt{n\left( \sigma_Z^2 + \zeta \right)} \right\} \;.\,
\label{Eq.Event_Q}
\end{align}
Then, observe that
\begin{align}
    \label{Eq.Error_Sum_1}
    P_{e,1}(i_1) + P_{e,2}(i_2,\mathbbmss{K}) & = \sup_{g\in\G} \left[ 1 - \int_{\sT_{\mathbbmss{K}}} f_{\fZ}(\fy-g\fc_{i_1}) \, d\fy \right]_{i_1 \in \mathbbmss{K}} + \sup_{g\in\G} \left[ \int_{\sT_{\mathbbmss{K}}} f_{\fZ}(\fy-g\fc_{i_2}) \, d\fy \right]_{i_2 \notin \mathbbmss{K}}
    \,.\,
\end{align}
Now consider the first integral in \eqref{Eq.Error_Sum_1} where for every $g\in\G$ we have,
\begin{align}
        \int_{\sT_{\mathbbmss{K}}} f_{\fZ}(\fy-g\fc_{i_1}) \, d\fy & \stackrel{(a)}{\leq} \int_{\underset{i_1\in \mathbbmss{K}}{\bigcup} \mathbbmss{T}_{i_1,g}} f_{\fZ}(\fy-g\fc_{i_1}) \, d\fy
        \nonumber\\&
        \stackrel{(a)}{=} \int_{\left( \underset{i_1\in \mathbbmss{K}}{\bigcup} \mathbbmss{T}_{i_1,g} \right) \cap \mathbbmss{P}_{i_1,i_2}} f_{\fZ}(\fy-g\fc_{i_1})\, \, d\fy + \int_{\left( \underset{i_1\in \mathbbmss{K}}{\bigcup} \mathbbmss{T}_{i_1,g} \right) \cap \mathbbmss{P}_{i_1,i_2}^c} f_{\fZ}(\fy-g\fc_{i_1})\, \, d\fy
        \nonumber\\&
        \stackrel{(b)}{\leq} \int_{\underset{i_1\in \mathbbmss{K}}{\bigcup} \left( \mathbbmss{T}_{i_1,g} \cap \mathbbmss{P}_{i_1,i_2} \right)} f_{\fZ}(\fy-g\fc_{i_1})\, \, d\fy + \int_{\underset{i_1\in \mathbbmss{K}}{\bigcup} \left( \mathbbmss{T}_{i_1,g} \cap \mathbbmss{P}_{i_1,i_2}^c \right)} f_{\fZ}(\fy-g\fc_{i_1})\, \, d\fy
        \nonumber\\
        & \stackrel{(c)}{\leq} \int_{\underset{i_1\in \mathbbmss{K}}{\bigcup} \mathbbmss{P}_{i_1,i_2}} f_{\fZ}(\fy-g\fc_{i_1})\, \, d\fy +  \int_{\underset{i_1\in \mathbbmss{K}}{\bigcup} \mathbbmss{Q}_{i_1,i_2}^c} f_{\fZ}(\fy-g\fc_{i_1}) d\fy \;,\,
        \label{Eq.LTP-T_i_1}
\end{align}
    where $(a)$ holds by the union bound, $(b)$ follows by the following
    \begin{subequations}
    \begin{align}
     \left( \underset{i_1\in \mathbbmss{K}}{\bigcup} \mathbbmss{T}_{i_1,g} \right) \cap \mathbbmss{P}_{i_1,i_2} \subset \underset{i_1\in \mathbbmss{K}}{\bigcup} \left( \mathbbmss{T}_{i_1,g} \cap \mathbbmss{P}_{i_1,i_2} \right) \;,\,
    \end{align}
    and
    \begin{align}
        \left( \underset{i_1\in \mathbbmss{K}}{\bigcup} \mathbbmss{T}_{i_1,g} \right) \cap \mathbbmss{P}_{i_1,i_2}^c \subset \underset{i_1\in \mathbbmss{K}}{\bigcup} \left( \mathbbmss{T}_{i_1,g} \cap \mathbbmss{P}_{i_1,i_2}^c \right) \;,\,
    \end{align}
    \end{subequations}
    and $(c)$ holds since
    \begin{align}
        \mathbbmss{Q}_{i_1,i_2}^c \supset \mathbbmss{T}_{i_1,g} \cap \mathbbmss{P}_{i_1,g}^c \;.\,
    \end{align}
Consider the second integral in \eqref{Eq.LTP-T_i_1}. Then, by the triangle inequality,
\begin{align}
    \norm{\fy-g\fc_{i,1}} & \geq
    \norm{\fy-g\fc_{i,2}}-\norm{g(\fc_{i,1}-\fc_{i,2})}
    \nonumber\\&
    = \norm{\fy-g\fc_{i,2}}-g\norm{\fc_{i,1}-\fc_{i,2}}
    \nonumber\\&
    > \sqrt{n(\sigma_Z^2+\zeta)}-g\norm{\fc_{i,1}-\fc_{i,2}}
    \nonumber\\&
    \geq\sqrt{n(\sigma_Z^2+\zeta)}-g\alpha_n \;.\,
\end{align}
    For sufficiently large $n$, this implies the following subset
        \begin{align}
           \mathbbmss{R}_{i_1,i_2}^c = \left\{y^n \in \mathbb{Y}^n \; : \, \norm{\fy - g\fc_{i,1}} > \sqrt{n\left( \sigma_Z^2 + \eta \right)} \right\} \;,\,
           \label{Eq.Regiong0_Slow}
        \end{align}
    for $\eta<\frac{\zeta}{2}$. That is,
    \begin{align}
        \left\{\fy \in \mathbb{Y}^n \; : \, \norm{\fy-g\fc_{i,2}} \geq
        \sqrt{n\left( \sigma_Z^2 + \zeta \right)} \right\} \quad \overset{\text{implies}}{\longrightarrow} \quad \left\{\fy \in \mathbb{Y}^n \; : \, \norm{\fy-g\fc_{i,1}} \geq
        \sqrt{n \left( \sigma_Z^2 + \eta \right)} \right\} \;.\,
    \end{align}
    Thus we deduce that
        \begin{align}
            \mathbbmss{R}_{i_1,i_2}^c \supset \mathbbmss{Q}_{i_1,i_2}^c \;,\,
        \end{align}
Hence, the second integral in the right hand side of \eqref{Eq.LTP-T_i_1} is bounded by
    \begin{align}
        \int_{\underset{i_1\in \mathbbmss{K}}{\bigcup} \mathbbmss{Q}_{i_1,i_2}^c} f_{\fZ}(\fy-g\fc_{i_1}) d\fy & \leq \int_{\underset{i_1\in \mathbbmss{K}}{\bigcup} \mathbbmss{R}_{i_1,i_2}^c} f_{\fZ}(\fy-g\fc_{i_1}) d\fy
        \nonumber\\&
        = \sum_{i_1 \in \mathbbmss{K}} \Pr\left( \norm{\fy-g\fc_{i,1}} \geq
        \sqrt{n \left( \sigma_Z^2 + \eta \right)} \right) 
        \nonumber\\
        & = K \cdot \Pr(\norm{\fZ}^2-n\sigma_Z^2>n\eta)
        \nonumber\\
        & \stackrel{(a)}{\leq} \frac{3\sigma_Z^4}{n^{1-\kappa} \eta^2}
        \nonumber\\
        & \leq \mu \;,\,
    \end{align}
    for sufficiently large $n$ with $\kappa \in [0,1)$, where $(a)$ holds by Chebyshev's inequality, followed by the substitution of $\fz\equiv \fy-g\fc_{i_1}$.
   Thus, by (\ref{Eq.LTP-T_i_1}),
    \begin{align}
        \label{Eq.ComplTypeISlow}
        \int_{\sT_{\mathbbmss{K}}} f_{\fZ}(\fy-g\fc_{i_1}) \, d\fy & \leq \int_{\underset{i_1\in \mathbbmss{K}}{\bigcup} \mathbbmss{T}_{i_1,g}} f_{\fZ}(\fy-g\fc_{i_1}) \, d\fy
        \nonumber\\&
        \leq \int_{\underset{i_1\in \mathbbmss{K}}{\bigcup} \mathbbmss{P}_{i_1,i_2}} f_{\fZ}(\fy-g\fc_{i_1}) \, d\fy + \mu \;.\,
    \end{align}
    Now, let us focus on the first integral in \eqref{Eq.LTP-T_i_1} with domain of $\mathbbmss{P}_{i_1,i_2}$, i.e., where
    \begin{align}
        \norm{\fy-g\fc_{i,2}}\leq\sqrt{n(\sigma_Z^2+\zeta)} \;.\,
        \label{Eq.ui2DistSlow}
    \end{align}
Observe that
    \begin{align}
        f_{\fZ}(\fy-g\fc_{i_1}) - f_{\fZ}(\fy-g\fc_{i_2})= f_{\fZ}(\fy-g\fc_{i_1})\left[1-e^{-\frac{1}{2\sigma_Z^2}\left(\norm{\fy-g\fc_{i_2}}^2-\norm{\fy-g\fc_{i_1}}^2\right)}\right] \;.\,
    \end{align}
    By the triangle inequality,
    \begin{align}
        \label{Ineq.triangleSlow}
        \norm{\fy-g\fc_{i_1}}\leq  \norm{\fy-g\fc_{i_2}} +  g\norm{\fc_{i_1} - \fc_{i_2}} \;.\,
    \end{align}
    Taking the square of both sides, we have
    \begin{align}
        \label{Ineq.Gaussian_Continuity}
        \norm{\fy-g\fc_{i_1}}^2 &\leq \norm{\fy-g\fc_{i_2}}^2 + g^2\norm{\fc_{i_2} - \fc_{i_1}}^2 +
        2\norm{\fy-g\fc_{i_2}}\cdot g\norm{\fc_{i_2} - \fc_{i_1}}
        \nonumber\\&
        \stackrel{(a)}{\leq} \norm{\fy-g\fc_{i_2}}^2+ g^2\alpha_n^2 + 2g\alpha_n\sqrt{n(\sigma_Z^2+\zeta)}
        \nonumber\\&
        \stackrel{(b)}{=} \norm{\fy-g\fc_{i_2}}^2 + \frac{4Ag^2}{n^{1+2(\kappa+b)}} +
        \frac{4g\sqrt{A(\sigma_Z^2+\zeta)}}{n^{\kappa + b}} \;,\,
    \end{align}
    where $(a)$ follows from (\ref{Ineq.Negated_Assump}) and (\ref{Eq.ui2DistSlow}), and $(b)$ holds by (\ref{Eq.Alpha_nSlow}). Now, in order to bound \eqref{Ineq.Gaussian_Continuity}, let us define,
    \begin{align}
        N_{\rm max} \overset{\text{\scriptsize def}}{=} 2\sigma_Z^2 \cdot \max \left( 4Ag^2 , 8g\sqrt{A(\sigma_Z^2+\zeta)} \right) \,.\,
    \end{align}
    Therefore, \eqref{Ineq.Gaussian_Continuity} is bounded as follows
    \begin{align}
        \label{Ineq.Norm_Squ_Diff}
        \norm{\fy-g\fc_{i_1}}^2-\norm{\fy-g\fc_{i_2}}^2 & \leq \frac{4Ag^2}{n^{1+2(\kappa+b)}} + \frac{4g\sqrt{A(\sigma_Z^2+\zeta)}}{n^{\kappa+b}}
        \nonumber\\&
        \leq \frac{2\sigma_Z^2 N_{\rm max}}{n^{\kappa+b}}
        \,,\,
    \end{align}
    where the last inequality holds since $n^{1+2(\kappa+b)} \geq n^{\kappa + b}$ for a given $\kappa$ and $b$, and every $n$. Now let us define
    \begin{align}
        \label{Eq.omega_n}
        \omega_n \overset{\text{\scriptsize def}}{=} \frac{N_{\rm max}}{n^{\kappa+b}} \,.\,
    \end{align}
    Then we employ inequality $1 - \frac{1}{x} \leq \ln x \,,\, \forall x > 0$ (\cite[Eq.~1]{Topse04}) by setting $x = \frac{1}{1-\omega_n}$ and provide an upper bound on $\omega_n$ as follows
    \begin{align}
        \omega_n & \leq \ln \left( \frac{1}{1 - \omega_n} \right)
        \nonumber\\&
        = \ln \left( \frac{n^{\kappa+b}}{n^{\kappa+b} - N_{\rm max}} \right)
        \,,\,
    \end{align}
    where conditions $x>0$ and $\omega_n < 1$ are fulfilled for sufficiently large $n$. Therefore by \eqref{Ineq.Norm_Squ_Diff} we obtain
    \begin{align}
        \norm{\fy-g\fc_{i_1}}^2-\norm{\fy-g\fc_{i_2}}^2 \leq 2\sigma_Z^2 \cdot \ln \left( \frac{n^{\kappa+b}}{n^{\kappa+b} - N_{\rm max}} \right)
        \,,\,
    \end{align}
    Hence,
    \begin{align}
        \label{Ineq.GaussianContinuitySlow}
        f_{\fZ}(\fy-g\fc_{i_1}) - f_{\fZ}(\fy-g\fc_{i_2}) & \leq f_{\fZ}(\fy-g\fc_{i_1}) \left(1-e^{-\frac{\omega_n}{2\sigma_Z^2}}\right)
        \nonumber\\&
        \leq f_{\fZ}\left(\fy-g\fc_{i_1}\right) \left(1-e^{- \ln \left( \frac{n^{\kappa+b}}{n^{\kappa+b} - N_{\rm max}} \right)} \right)
        \nonumber\\&
        \leq f_{\fZ}\left(\fy-g\fc_{i_1}\right) \left(1 - \frac{n^{\kappa+b} - N_{\rm max}}{n^{\kappa+b}} \right)
        \nonumber\\&
        \leq f_{\fZ}\left(\fy-g\fc_{i_1}\right) \cdot \frac{N_{\rm max}}{n^{\kappa+b}}
        \nonumber\\&
        = f_{\fZ}\left(\fy-g\fc_{i_1}\right) \cdot \omega_n \;,\,
    \end{align}
    Now we obtain,
\begin{align}
      \label{Ineq.Conv_1}
      e_1 + e_2 & \geq P_{e,1}(i_1) + P_{e,2}(i_2,\mathbbmss{K})
      \nonumber\\
      & \stackrel{(a)}{\geq} \sup_{g\in\G} \left[ P_{e,1}(i_1 | g) \right] + \sup_{g\in\G} \left[ P_{e,2}(i_2,\mathbbmss{K} | g) \right]
      \nonumber\\
      & \stackrel{(b)}{\geq} \sup_{g\in\G} \left[ P_{e,1}(i_1 | g) + P_{e,2}(i_2,\mathbbmss{K} | g) \right]
      \nonumber\\
      & \stackrel{(c)}{=} \sup_{g\in\G} \left[ 1 - \int_{\sT_{\mathbbmss{K}}} f_{\fZ}(\fy-g\fc_{i_1}) \, d \mathbf{y} + \int_{\sT_{\mathbbmss{K}}} f_{\fZ}(\fy-g\fc_{i_2}) \, d\fy \right]
\end{align}
where $(a)$ follows by \eqref{Eq.TypeIError-Sup} and \eqref{Eq.TypeIIError-Sup}, $(b)$ holds since supremum is sub-additive and $(c)$ is due to definitions of error in \eqref{Eq.TypeIError} and \eqref{Eq.TypeIIError}. Now we proceed to bound \eqref{Ineq.Conv_1} as follows
\begin{align}
    \label{Ineq.Conv_2}
    & \sup_{g\in\G} \left[ 1 - \int_{\sT_{\mathbbmss{K}}} f_{\fZ}(\fy-g\fc_{i_1}) \, d \mathbf{y} + \int_{\sT_{\mathbbmss{K}}} f_{\fZ}(\fy-g\fc_{i_2}) \, d\fy \right]
    \nonumber\\&
    \stackrel{(a)}{\geq} \sup_{g\in\G} \left[ 1 - \mu - \int_{\underset{i_1\in \mathbbmss{K}}{\bigcup} \mathbbmss{P}_{i_1,i_2}} f_{\fZ}(\fy-g\fc_{i_1}) \, d\fy +
    \int_{\underset{i_1\in \mathbbmss{K}}{\bigcup} \mathbbmss{T}_{i_1,g}} f_{\fZ}(\fy-g\fc_{i_2}) \, d\fy \right]
    \nonumber\\&
    \stackrel{(b)}{\geq} \sup_{g\in\G} \left[ 1 - \mu - \int_{\underset{i_1\in \mathbbmss{K}}{\bigcup} \mathbbmss{P}_{i_1,i_2}} f_{\fZ}(\fy-g\fc_{i_1}) \, d\fy + \int_{\underset{i_1\in \mathbbmss{K}}{\bigcup} \mathbbmss{P}_{i_1,i_2}} f_{\fZ}(\fy-g\fc_{i_2}) \, d\fy \right]
    \nonumber\\&
    \stackrel{(c)}{=} \sup_{g\in\G} \left[ 1 - \mu - \int_{\underset{i_1\in \mathbbmss{K}}{\bigcup} \mathbbmss{P}_{i_1,i_2}} \left[ f_{\fZ}(\fy-g\fc_{i_1}) - f_{\fZ}(\fy-g\fc_{i_2}) \right] \, d\fy \right]
    \nonumber\\&
\end{align}
where $(a)$ holds by \eqref{Eq.ComplTypeISlow} and $(b)$ follows from $\mathbbmss{P}_{i_1,i_2} \subset \mathbbmss{T}_{i_1,g}$. Now we proceed to bound \eqref{Ineq.Conv_2} as follows
\begin{align}
    \label{Ineq.Conv_3}
    & \sup_{g\in\G} \left[ 1 - \mu - \int_{\underset{i_1\in \mathbbmss{K}}{\bigcup} \mathbbmss{P}_{i_1,i_2}} \left[ f_{\fZ}(\fy-g\fc_{i_1}) - f_{\fZ}(\fy-g\fc_{i_2}) \right] \, d\fy \right]
    \nonumber\\&
    \stackrel{(a)}{\geq} \sup_{g\in\G} \left[ 1 - \mu - \omega_n \int_{\underset{i_1\in \mathbbmss{K}}{\bigcup} \mathbbmss{P}_{i_1,i_2}} f_{\fZ}(\fy-g\fc_{i_1}) \, d\fy \right]
    \nonumber\\&
    \stackrel{(b)}{\geq} \sup_{g\in\G} \left[ 1 - \mu - \omega_n \sum_{i_1 \in \mathbbmss{K}}
    \int_{\mathbbmss{P}_{i_1,i_2}}f_{\fZ} (\fy - g \fc_{i_1}) \, d\fy \right]
    \nonumber\\&
    \stackrel{(c)}{\geq} \sup_{g\in\G} \left[ 1 - \mu - \omega_n \cdot |\mathbbmss{K}|
    \right]
    \nonumber\\&
    \stackrel{(d)}{=} \sup_{g\in\G} \left[ 1 - \mu - \frac{KN_{\rm max}}{n^{b+\kappa}}
      \right]
      \nonumber\\&
    \stackrel{(e)}{\geq} \sup_{g\in\G} \left[ 1 - \mu - \pi \right]
    \nonumber\\&
    = 1 - 2\mu - \pi \,,\, 
\end{align}
where $(a)$ follows by \eqref{Ineq.GaussianContinuitySlow}, $(b)$ holds by the union bound, $(c)$ follows from
\begin{align}
    \int_{\mathbbmss{P}_{i_1,i_2}} f_{\fZ} ( \fy - g \fc_{i_1} ) \, d\fy & = \Pr\left( \norm{\fy - g\fc_{i_1}} \leq \sqrt{n\left( \sigma_Z^2 + \zeta \right)} \right) \leq 1 \,,\,
\end{align}
and $(c)$ follows since $|\mathbbmss{K}| = K = n^{\kappa}$, $(d)$ follows from \eqref{Eq.omega_n}, and $(e)$ holds since $\frac{KN_{\rm max}}{n^{b+\kappa}} = \frac{1}{n^b} \leq \pi$ for sufficiently large $n$. Thereby, recalling \eqref{Ineq.Conv_1},\eqref{Ineq.Conv_2},\eqref{Ineq.Conv_3} we obtain
\begin{align}
    e_1 + e_2 \geq 1 - 2\mu - \pi \,.\,
\end{align}

Clearly, this is a contradiction since the error probabilities tend to zero as $n\rightarrow\infty$. Thus, the assumption in (\ref{Ineq.Negated_Assump}) is false. This completes the proof of Lemma~\ref{Lem.DConv}.
\printendnotes
\bibliography{Lit}
\end{document}